\documentclass[journal]{IEEEtran}

\ifCLASSINFOpdf
\else
\fi

\usepackage{setspace}
\usepackage{bbm}
\usepackage[cmex10]{amsmath}
\usepackage{amssymb}
\usepackage{cite}
\usepackage{graphicx}
\usepackage{array,color}
\usepackage{amsmath}
\allowdisplaybreaks
\usepackage{stfloats}
\usepackage{graphicx}
\usepackage{epstopdf}
\usepackage{subfigure}
\usepackage{tabularx}
\usepackage{epsfig,epsf,color,balance,cite}
\usepackage{algorithmic}
\usepackage{algorithm}
\usepackage{url}
\usepackage{bm}
\usepackage{multirow}

\usepackage{amsthm}

\newtheorem{theorem}{Theorem}
\newtheorem{lemma}[theorem]{Lemma}

\ifodd 0
\usepackage{soul}
\usepackage{color}
\setstcolor{red}

\newcommand{\rev}[1]{{\color{red}#1}} 
\newcommand{\del}[1]{\st{#1}} 

\newcommand{\com}[1]{\textbf{\color{red} (COMMENT: #1)}} 
\newcommand{\response}[1]{\textbf{\color{green} (RESPONSE: #1)}} 
\else

\newcommand{\rev}[1]{#1}

\newcommand{\del}[1]{}

\newcommand{\com}[1]{}
\newcommand{\comg}[1]{}
\newcommand{\response}[1]{}
\fi

\title{\huge {Efficient Channel Estimation for Double-IRS Aided Multi-User MIMO System}}

\author{ 
	Beixiong Zheng,~\IEEEmembership{Member,~IEEE}, Changsheng You,~\IEEEmembership{Member,~IEEE}, and Rui Zhang,~\IEEEmembership{Fellow,~IEEE} 
	
	\thanks{
		Part of this work will be presented in IEEE International Conference on Communications (ICC), Montreal, Canada, 2021~\cite{zheng2020Uplink}. This work was supported by the National University of Singapore under Grant R-261-518-005-720. \emph{(Corresponding author: Changsheng You.)}
		
		The authors are with the Department of Electrical and Computer Engineering, National University of Singapore, Singapore 117583,
		email: \{elezbe, eleyouc, elezhang\}@nus.edu.sg.

	}
}

\begin{document}
\markboth{IEEE Transactions on Communications, Vol. XX, No. XX, XXX 2021}{SKM: My IEEE article}
\maketitle
\begin{abstract}
To achieve the more significant passive beamforming gain in the double-intelligent reflecting surface (IRS) aided system over the conventional single-IRS counterpart,  channel state information (CSI) is indispensable in practice but also more challenging to acquire, due to the presence of not only the single- but also double-reflection links that are intricately coupled and also entail more channel coefficients for estimation.
In this paper, we propose a new and efficient channel estimation scheme for the double-IRS aided multi-user multiple-input multiple-output (MIMO) communication system to resolve the cascaded CSI of both its single- and double-reflection links. 
First, for the single-user case, the single- and double-reflection channels are efficiently estimated at the multi-antenna base station (BS) with both the IRSs turned ON (for maximal signal reflection), by exploiting the fact that their cascaded channel coefficients are scaled versions of their superimposed lower-dimensional CSI.
Then, the proposed channel estimation scheme is extended to the multi-user case, where given an arbitrary user's cascaded channel (estimated as in the single-user case), the other users' cascaded channels can also be expressed as lower-dimensional scaled versions of it and thus efficiently estimated at the BS.
Simulation results verify the effectiveness of the proposed channel estimation scheme and joint training reflection design for double IRSs, as compared to various benchmark schemes.
\end{abstract}
\IEEEpeerreviewmaketitle

\section{Introduction}
\IEEEPARstart{I}{ntelligent}
reflecting surface (IRS) has recently emerged as an innovative solution to the realization of smart
and reconfigurable environment for wireless communications \cite{wu2020intelligent,qingqing2019towards,Renzo2019Smart,basar2019wireless}. Specifically, IRS consists of a large number of passive reflecting 
elements with ultra-low power consumption, each being able to control the reflection phase shift and/or amplitude of the incident signal in a programmable manner so as to collaboratively reshape the wireless propagation channel in favor of signal transmission. 
As such, different from the existing transmission techniques that can only adapt to but have no control over the random wireless channels,
IRS provides a new means to effectively combat with wireless channel fading impairment and co-channel interference.
Furthermore, as being of light weight and free of radio frequency (RF) chains, large-scale IRSs can
be densely deployed in various wireless communication systems \cite{Zheng2020IRSNOMA,Pan2020Multicell,Yang2020IRS,Wu2020Weighted2,Pan2020Intelligent,Cui2019Secure,Zhou2020Delay} to enhance their performance at low as well as sustainable energy and hardware cost.

\rev{Prior works on IRS mainly considered the wireless communication systems aided by one or multiple distributed IRSs (see, e.g., \cite{Wu2019TWC,Huang2019Reconfigurable,wu2019beamforming,mei2020performance,Zhang2020Intelligent,yang2020energy,Zhang2020Capacity,li2019jointactive,yang2020outage,Gao2020Distributed,Wei2020Sum}),}
each independently serving its nearby users without taking into account the inter-IRS signal reflection, which, however, fails to capture the cooperative beamforming gains between IRSs to further improve the system performance. Only recently,
the cooperative beamforming over inter-IRS channel has been explored in the double-IRS aided communication system \cite{Han2020Cooperative,Zheng2020DoubleIRS,you2020wireless}, which was shown to achieve a much higher-order passive beamforming gain than that of the conventional single-IRS aided system (i.e., ${\cal O}(M^4)$ versus ${\cal O}(M^2)$ with $M$ denoting the total number of reflecting elements/subsurfaces in both systems \cite{Han2020Cooperative}). However, achieving this more appealing passive beamforming gain requires more channel training overhead in practice, due to the significantly more channel coefficients to be estimated over the inter-IRS double-reflection link, in addition to the single-reflection links in the conventional single-IRS system.
\rev{Note that existing works on IRS channel estimation mainly focused on the channel state information (CSI) acquisition for the single-reflection links only \cite{zheng2019intelligent,zheng2020intelligent,zheng2020fast,you2019progressive,you2020fast,wang2019channel,jensen2019optimal,liu2019matrix,wang2019compressed,Zhang2020Cascaded},} which, however, are inapplicable to the double-IRS aided system with the co-existence of single- and double-reflection links as shown in Fig.~\ref{system}.
In \cite{Han2020Cooperative}, the authors assumed that the two IRSs in the double-IRS system are equipped with receive RF chains to enable
signal sensing for estimating their channels with the base station (BS)/user, separately. However, even with receive RF chains integrated to IRSs, \rev{the channel estimation for the inter-IRS (i.e., IRS~1$\rightarrow$IRS~2 in Fig.\ref{system}) link still remains a difficult task in practice (unless active sensors that can both transmit and receive signals are mounted on the IRSs, which inevitably incurs additional hardware cost and complexity).}
In contrast, the double-IRS channel estimation with fully-passive IRSs was investigated in \cite{you2020wireless}, but without the single-reflection links considered and for the single-user case only. 
In \cite{zheng2020Uplink}, we presented a decoupled channel estimation scheme for the general double-IRS aided system with multiple BS antennas/users, which successively estimates the two single-reflection channels (each corresponding to one of the two IRSs turned ON with the other turned OFF, respectively) and then the double-reflection channel (with the signals canceled over the two single-reflection channels estimated) in a decoupled manner. 
\rev{Although this scheme substantially reduces the training overhead of that in \cite{you2020wireless}, it suffers from not only
residual interference due to imperfect signal cancellation (based on the estimated single-reflection channels) but also
considerable reflection power loss due to the ON/OFF reflection control of the two IRSs during the channel training, which thus degrade the channel estimation accuracy.}

\begin{figure}[!t]
	\centering
	\includegraphics[width=3.0in]{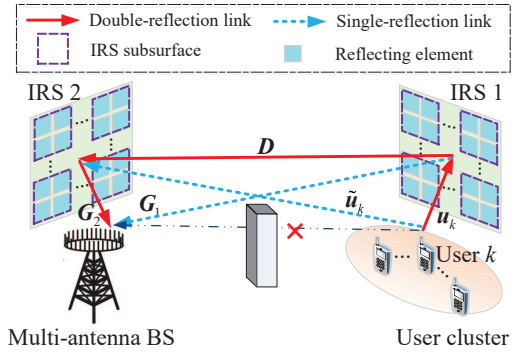}
	\caption{A double-IRS cooperatively aided multi-user MIMO communication system in the uplink.}
	\label{system}
\end{figure}
To overcome the above issues, we propose in this paper a new and efficient channel estimation scheme for the double-IRS aided multi-user multiple-input multiple-output (MIMO) communication system shown in Fig.~\ref{system}, 
where two fully-passive IRSs are deployed near a multi-antenna BS and a cluster of $K$ users, respectively, to assist their communications.
\rev{Different from the decoupled channel estimation scheme in \cite{zheng2020Uplink}, we propose to jointly estimate the cascaded CSI of the single- and double-reflection links with the always-ON training reflection (i.e., all the reflecting elements/subsurfaces of the two IRSs are switched ON with the maximum reflection amplitude during the entire channel training) while properly tuning their reflection phase-shifts over time, which thus
achieves the full-reflection power for improving the channel estimation accuracy significantly.}
Note that with the always-ON reflection of both IRSs, the two single-reflection channels (each corresponding to one of the two IRSs, respectively) are superimposed with the double-reflection channel at all time, which makes them difficult
to be estimated separately at the BS based on the received user pilot signals over the  intricately coupled single- and double-reflection links.
To tackle this difficulty, we explore the intrinsic relationship between the single- and double-reflection channels as well as that among different users 
and design the training reflection phase-shifts of the two IRSs
accordingly 
to achieve low training overhead and yet high channel estimation accuracy. Specifically, the main contributions of this paper are summarized as follows (please refer to Fig.~\ref{system} for the channel illustration).
\begin{itemize}
	\item First, for the single-user case, the superimposed CSI of the single- and double-reflection channels related to IRS~2 is estimated at the multi-antenna BS with dynamically tuned reflection phase-shifts of IRS~2 over time and those of IRS 1 being fixed.
	Then, the high-dimensional double-reflection (i.e., user$\rightarrow$IRS~1$\rightarrow$IRS~2$\rightarrow$BS) channel and the single-reflection (i.e., user$\rightarrow$IRS~2$\rightarrow$BS) channel are efficiently estimated,
	by exploiting the fact that their cascaded channel coefficients are scaled versions of their superimposed CSI due to their commonly shared IRS~2$\rightarrow$BS link; as a result, by taking the superimposed CSI as the reference CSI, only the lower-dimensional scaling factors and the other single-reflection (i.e., user$\rightarrow$IRS~1$\rightarrow$BS) channel need to be estimated at the multi-antenna BS, thus greatly reducing the training overhead. The proposed channel estimation scheme is further extended to the multi-user case, where given an arbitrary user's cascaded channel estimated as in the single-user case, the other users' cascaded channels are also (lower-dimensional) scaled versions of it and thus can be efficiently estimated at the BS with low training overhead.
	\item Moreover, for the proposed channel estimation scheme, we jointly optimize the training reflection phase-shifts of the two IRSs to minimize the channel estimation error over different training phases, for which some closed-form optimal solutions are derived with desired orthogonality.
	In addition, we analytically characterize the minimum training overhead of the proposed scheme, which is shown to decrease with the increasing number of BS antennas $N$ until reaching its lower bound with $N>M$.
	We also show an interesting trade-off in training time allocation over different training phases with the intricate error propagation effect taken into account.
	\item Last, we provide extensive numerical results to validate the superiority of our proposed channel estimation scheme to other benchmark schemes in terms of both training overhead and channel estimation error.
	Besides, the effectiveness of the proposed joint training reflection phase-shift design for double IRSs is also corroborated
	by simulation results, as compared to other heuristic training reflection designs.
\end{itemize}

The rest of this paper is organized as follows. Section~\ref{sys} presents the system model for the double-IRS aided multi-user MIMO system. In Sections~\ref{SU_design_ON} and \ref{MU_design_ON}, we propose efficient channel estimation schemes with optimized training reflection phase-shift designs for the single-user and multi-user cases, respectively. 
Simulation results are presented in Section \ref{Sim} to evaluate the performance of the
proposed channel estimation scheme and training reflection phase-shift design. Finally, conclusions are drawn in Section~\ref{conlusion}.

\emph{Notation:} 
Upper-case and lower-case boldface letters denote matrices and column vectors, respectively.
Superscripts ${\left(\cdot\right)}^{T}$, ${\left(\cdot\right)}^{H}$, ${\left(\cdot\right)}^{*}$ ${\left(\cdot\right)}^{-1}$, and ${\left(\cdot\right)}^{\dagger}$ stand for the transpose, Hermitian transpose, conjugate, matrix inversion, and pseudo-inverse operations, respectively.
${\mathbb C}^{a\times b}$ denotes the space of ${a\times b}$ complex-valued matrices.
For a complex-valued vector $\bm{x}$, $\lVert\bm{x}\rVert$ denotes its $\ell_2 $-norm 
and $\mathrm{diag} (\bm{x})$ returns a diagonal matrix with the elements in $\bm{x}$ on its main diagonal. 
For a general matrix ${\bm A}$,
${\rm rank} \left( {\bm A} \right)$ returns its rank,
$\left\|  {\bm A}\cdot \right\|_F$ denotes its Frobenius norm,
${\rm vec} \left({\bm A} \right)$ denotes the vectorization of ${\bm A}$ by stacking its columns into a column vector,
and $[{\bm A}]_{a:b,c:d}$ denotes the submatrix of ${\bm A}$ with its rows from $a$ to $b$ and columns from $c$ to $d$.
${\cal O}(\cdot)$ denotes the standard big-O notation,
$\lceil \cdot \rceil$ is the ceiling function,
$\otimes$ denotes the Kronecker product, 
and ${\mathbb E}\{\cdot\}$ stands for the statistical expectation.
${\bm I}$, ${\bm 1}$, and ${\bm 0}$ denote an identity matrix, an all-one vector/matrix, and an all-zero vector/matrix, respectively, with appropriate dimensions.
The distribution of a circularly symmetric complex Gaussian (CSCG) random vector with mean vector $\bm \mu$ and covariance matrix ${\bm \Sigma}$ is denoted by ${\mathcal N_c }({\bm \mu}, {\bm \Sigma} )$; and $\sim$ stands for ``distributed as".
\section{System Model}\label{sys}
Consider a double-IRS cooperatively aided multi-user MIMO communication system shown in Fig.~\ref{system}, in which two distributed IRSs (referred to as IRS~1 and IRS~2) are deployed to assist the communications between a cluster of $K$ single-antenna users and an $N$-antenna BS.\footnote{\rev{The results in this paper can be readily extended to the case of multiple user clusters each with a helping IRS deployed locally, by e.g., treating these distributed IRSs as an equivalent IRS (i.e., IRS~1) of larger size accordingly.}}
As in \cite{Zheng2020DoubleIRS}, we consider a practical scenario where the direct links between the users in the cluster of interest and the BS are severely blocked due to obstacles (e.g.,
walls/corners in the indoor environment) and thus can be ignored.\footnote{
	\rev{If the direct links are non-negligible, the BS can estimate them by using the conventional channel estimation method with orthogonal/sequential pilots sent from the users and the two IRSs both turned OFF.}}
To bypass the blockage as well as minimize the path loss, IRSs~1 and 2 are properly placed near the cluster of users and the BS, respectively,
such that the $K$ users can be effectively served by the BS through both the single- and double-reflection links created by them.
Let $M$ denote our budget on the total number of passive subsurfaces for the two distributed IRSs, where IRSs~1 and 2 comprise $M_1$ and $M_2$ subsurfaces, respectively, with $M_1+M_2=M$.
Similar to the element-grouping strategy in \cite{yang2019intelligent,zheng2019intelligent}, each of these IRS subsurfaces is composed of an arbitrary number of adjacent reflecting elements that share a common phase shift for reducing the channel estimation and reflection design
complexity.\footnote{\rev{In this paper, we consider the IRS element-based product-distance path-loss model \cite{wu2020intelligent}, where each IRS element is assumed to be located sufficiently far from the transmitters/receivers (i.e., with the distance larger than several wavelengths, which usually holds in practice) and thus the far-field propagation condition is valid. Also note that without changing the IRS hardware or path-loss model, the resultant subsurface-wise cascaded channels can be obtained from the element-wise cascaded channels with proper aggregation, as shown in \cite{you2020wireless}.}}
In this paper, we assume the quasi-static flat-fading channel model for all channels involved, which remain approximately constant during each channel coherence interval.

Let ${{\bm u}}_{k}\triangleq\left[u_{k,1},\ldots, u_{k,M_1}\right]^T\in {\mathbb{C}^{M_1\times 1}}$, ${\tilde{\bm u}}_{k}\triangleq\left[{\tilde u}_{k,1},\ldots, {\tilde u}_{k,M_2}\right]^T\in {\mathbb{C}^{M_2\times 1}}$, ${{\bm D}}\triangleq\left[{\bm d}_{1},\ldots, {\bm d}_{M_1}\right] \in {\mathbb{C}^{M_2\times M_1 }}$, ${{\bm G}}_1 \in {\mathbb{C}^{N\times M_1 }}$,
and ${{\bm G}}_2 \in {\mathbb{C}^{N\times M_2 }}$ 
denote the baseband equivalent channels in the uplink for the user~$k$$\rightarrow$IRS~1, user~$k$$\rightarrow$IRS~2, IRS~1$\rightarrow$IRS~2, IRS~1$\rightarrow$BS, and IRS~2$\rightarrow$BS links, respectively, with $k=1,\ldots,K$. 
Let ${\bm \theta}_\mu\triangleq[{\theta_{\mu,1}},\ldots,{\theta_{\mu,M_\mu}}]^T=\left[ \beta_{\mu,1}e^{j \phi_{\mu,1}},  \ldots, \beta_{\mu,M_\mu} e^{j \phi_{\mu,M_\mu}}\right]^T\in {\mathbb{C}^{M_\mu\times 1}}$ denote the equivalent reflection coefficients of IRS $\mu$ with $\mu\in \{1,2\}$, where $\beta_{\mu,m} \in [0, 1]$ and $\phi_{\mu,m} \in [0, 2\pi)$ are
the reflection amplitude and phase shift of subsurface~$m$ at IRS $\mu$, respectively.
Thus, the effective channel from user~$k$ to the BS is the superimposition of the double-reflection link\footnote{\rev{Although there exists another double-reflection link over the user~$k$$\rightarrow$IRS~2$\rightarrow$IRS~1$\rightarrow$BS channel, it suffers from much higher path loss due to the much longer propagation distance (see Fig.~\ref{system}) and thus is ignored in this paper.}} and the two single-reflection links (see Fig.~\ref{system}), which is given by
\begin{align}
	{\bm h}_k=&{{\bm G}}_2 {\bm \Phi}_2 {{\bm D}} {\bm \Phi}_1 {{\bm u}}_{k} 
	+{{\bm G}}_2 {\bm \Phi}_2 {\tilde{\bm u}}_{k} 
	+ {{\bm G}}_1 {\bm \Phi}_1 {{\bm u}}_{k} \label{superposed0}
\end{align}
where ${\bm \Phi}_\mu={\rm diag} \left( {\bm \theta}_\mu \right)$ denotes the diagonal reflection matrix of IRS $\mu$ with $\mu\in \{1,2\}$.
With fully-passive IRSs 1 and 2, it is infeasible to acquire the CSI between the two IRSs as well as that with the BS/users separately. Nonetheless, it was shown in \cite{Zheng2020DoubleIRS} that the cascaded CSI (to be specified below) is sufficient for designing the cooperative reflection/passive beamforming at the two IRSs to maximize the data transmission rate without loss of optimality.
As such, let ${{\bm R}}_{k}={{\bm G}}_1 {\rm diag} \left( {{\bm u}}_{k} \right)\in {\mathbb{C}^{N\times M_1 }}$ (${\tilde{\bm R}}_{k}={{\bm G}}_2   {\rm diag} \left( {\tilde{\bm u}}_{k} \right)\in {\mathbb{C}^{N\times M_2 }}$) denote the cascaded user~$k$$\rightarrow$IRS~1 (IRS~2)$\rightarrow$BS channel (without taking the effect of IRS reflection yet) for the single-reflection link reflected by IRS~1 (IRS~2), and ${\tilde{\bm D}}_k\triangleq\left[{\tilde{\bm d}}_{k,1},\ldots, {\tilde{\bm d}}_{k,M_1}\right]={{\bm D}} {\rm diag} \left( {{\bm u}}_{k} \right)\in {\mathbb{C}^{M_2\times M_1 }}$ denote the cascaded user~$k$$\rightarrow$IRS~1$\rightarrow$ IRS~2 channel (without IRS reflection) with \rev{${\tilde{\bm d}}_{k,m}={\bm d}_{m} {{u}}_{k,m}\in {\mathbb{C}^{M_2\times 1 }}, \forall m=1,\ldots,M_1$.}
Then, the channel model in \eqref{superposed0} can be equivalently expressed as
\begin{align}
	{\bm h}_k
	&={{\bm G}}_2 {\bm \Phi}_2  {\tilde{\bm D}}_k  {\bm \theta}_1
	+ {\tilde{\bm R}}_{k} {\bm \theta}_2
	+ {{\bm R}}_{k}{\bm \theta}_1  \notag\\
	&={{\bm G}}_2 \left[{\bm \Phi}_2{\tilde{\bm d}}_{k,1},\ldots, {\bm \Phi}_2{\tilde{\bm d}}_{k,M_1} \right] {\bm \theta}_1+ {\tilde{\bm R}}_{k} {\bm \theta}_2
	+ {{\bm R}}_{k}{\bm \theta}_1 \notag\\
	&={{\bm G}}_2 \left[ {\rm diag} \left({\tilde{\bm d}}_{k,1}\right){\bm \theta}_2,\ldots, {\rm diag} \left({\tilde{\bm d}}_{k,M_1}\right){\bm \theta}_2 \right] {{\bm \theta}}_1\notag\\
	&~~~+ {\tilde{\bm R}}_{k}{\bm \theta}_2
	+ {{\bm R}}_{k}{\bm \theta}_1 \notag\\
	&=\sum\limits_{m=1}^{M_1}  \underbrace{{{\bm G}}_2 ~{\rm diag} \left({\tilde{\bm d}}_{k,m}\right)}_{{{\bm Q}}_{k,m}} {\bm \theta}_2 {\theta}_{1,m}
	+ {\tilde{\bm R}}_{k}{\bm \theta}_2
	+ {{\bm R}}_{k}{\bm \theta}_1
	\label{superposed3}
\end{align}
where ${{\bm Q}}_{k,m}\in {\mathbb{C}^{N\times M_2 }}$ denotes the cascaded user~$k$$\rightarrow$IRS~1$\rightarrow$IRS~2$\rightarrow$BS channel associated with subsurface $m$ of IRS~1, $\forall m=1,\ldots,M_1$ for the double-reflection link.
\rev{According to
\eqref{superposed3}, it is sufficient to acquire the knowledge of the cascaded channels $\left\{ {{\bm Q}}_{k,m} \right\}_{m=1}^{M_1}$, ${\tilde{\bm R}}_{k}$, and ${{\bm R}}_{k}$ for jointly designing the passive beamforming $\left\{{\bm \theta}_1, {\bm \theta}_2\right\}$ for uplink/downlink data transmission in the double-IRS aided multi-user MIMO system, which is applicable to, but not limited to, the max-min SINR problem considered in \cite{Zheng2020DoubleIRS}.}
However, the total number of channel coefficients in the cascaded channels $\left\{ {{\bm Q}}_{k,m} \right\}_{m=1}^{M_1}$, ${\tilde{\bm R}}_{k}$, and ${{\bm R}}_{k}$ is prohibitively large, which consists of two parts:
\begin{itemize}
	\item The number of channel coefficients (equal to $K \times NM_1M_2$) for the high-dimensional double-reflection link (i.e., $\left\{ {{\bm Q}}_{k,m} \right\}_{m=1}^{M_1}$), which are newly introduced due to the double IRSs.
	\item  The number of channel coefficients (equal to $K \times N(M_1+M_2)$) for the two single-reflection links (i.e., ${{\bm R}}_{k}$ and ${\tilde{\bm R}}_{k}$), which exist in the conventional single-IRS aided system (with either IRS 1 or IRS 2 present).
\end{itemize}
As such, 
the number of channel coefficients for the double-reflection link is of much higher-order than that for 
the two single-reflection links due to the fact that $M_1M_2\gg M_1+M_2$ given large $M_1$ and $M_2$ in practice, which makes the channel estimation problem more challenging for the double-IRS aided system, as compared to the single-IRS counterpart.
Note that given a limited channel coherence interval, such a considerably larger number of channel coefficients may require  significantly more training overhead that renders much less or even no time for data transmission, thus resulting in reduced achievable rate for the double-IRS aided system (despite the higher-order passive beamforming gain over the double-reflection link assuming perfect CSI as shown in \cite{Zheng2020DoubleIRS}).
Moreover, the cascaded CSI of the double-refection link needs to be efficiently estimated along with that of the two single-reflection links. 
However, the pilot signals from the users are intricately coupled over the single- and double-reflection links when both the two IRSs are turned ON for maximizing the reflected signal power, which thus calls for new designs for their joint channel estimation at the multi-antenna BS.

 To tackle the above challenges, we propose a new and efficient channel estimation scheme for the considered double-IRS aided system that achieves practically low training overhead and yet high channel estimation accuracy.\footnote{\rev{Note that the proposed channel estimation scheme is also applicable to the multi-IRS aided multiuser system considered in \cite{you2020deploy}, where a multi-antenna BS serves multiple users with the aid of two types of IRSs, which are deployed near the BS and distributed users, and thus can be treated as two equivalent (aggregated) IRSs (i.e., IRS 2 and IRS 1, respectively) of larger size accordingly.}}
 We first consider the single-user setup, i.e., $K=1$, to illustrate the main idea of the proposed channel estimation scheme and training reflection design in Section~\ref{SU_design_ON}, and then extend the results to the general multi-user case in Section~\ref{MU_design_ON}.  
 To reduce the hardware cost and maximize the signal reflection power, we consider the always-ON training reflection design for which all the reflecting elements/sub-surfaces at the two IRSs are turned ON during the entire channel training and their reflection amplitudes are set to the maximum value of one  
 (i.e., $\beta_{\mu,m}=1, \forall  m=1,\ldots,M_\mu,\mu\in \{1,2\} $), and thereby focus on the dynamically tuned reflection phase-shift design for achieving efficient joint channel estimation of the single- and double-reflection links.

\section{Double-IRS Channel Estimation for Single-User Case}\label{SU_design_ON} 
 In this section, we study the cascaded channel estimation and training design for the single-user case with $K=1$. For notational convenience, the user index $k$ is omitted in this section.

\subsection{Proposed Channel Estimation Scheme}\label{ProposedCE}
Note that with both IRSs~1 and 2 turned ON simultaneously, the double-reflection channel $\left\{ {{\bm Q}}_{m} \right\}_{m=1}^{M_1}$ and two single-reflection channels $\left\{{{\bm R}}, {\tilde{\bm R}}\right\}$ are superimposed in \eqref{superposed3}, which are difficult to be estimated separately at the BS as the received pilot signals from the user are intricately coupled over the single- and double-reflection links.
To tackle this difficulty, we propose an efficient channel estimation scheme, which includes two phases as elaborated below.

{\bf Phase I:} 
In this phase, we fix the reflection phase-shifts of IRS~1, i.e., ${\bm \theta}_{1,\rm I}^{(i)}={\bm \theta}_{1,\rm I}, \forall i=1,\ldots,I_1$ and dynamically tune the reflection phase-shifts of IRS~2, i.e., ${\bm \theta}_{2,\rm I}^{(i)}$, over different pilot symbols to facilitate the
cascaded channel estimation, where $I_1$ denotes the number of pilot symbols in Phase I.
For simplicity, we set ${\bm \theta}_{1,\rm I}={\bm 1}_{M_1\times 1}$ and thus the effective channel in \eqref{superposed3} during symbol~$i$ of Phase~I, denoted by ${\bm h}_{\rm I}^{(i)}$, is given by\footnote{Note that ${\bm \theta}_{1,\rm I}$ can be arbitrarily set in Phase I for our proposed channel estimation scheme without affecting its performance.}
\begin{align}\label{superposed5}
{\bm h}_{\rm I}^{(i)}
&=\sum_{m=1}^{M_1}  {{\bm Q}}_{m}{\bm \theta}_{2,\rm I}^{(i)}+{\tilde{\bm R}}{\bm \theta}_{2,\rm I}^{(i)}+ {{\bm G}}_1 {{\bm u}}\notag\\
&=\sum\limits_{m=1}^{M_1}  {{\bm G}}_2 ~{\rm diag} \left({\tilde{\bm d}}_{m}\right) {\bm \theta}_{2,\rm I}^{(i)}
+ {{\bm G}}_2 {\rm diag} \left( {\tilde{\bm u}} \right){\bm \theta}_{2,\rm I}^{(i)}
+ {{\bm G}}_1 {{\bm u}}\notag\\
&={{\bm G}}_2{\rm diag}\hspace{-0.1cm} \left(\hspace{-0.1cm}{\tilde{\bm u}}+\sum_{m=1}^{M_1}{\tilde{\bm d}}_{m}\hspace{-0.1cm} \right) \hspace{-0.1cm}{\bm \theta}_{2,\rm I}^{(i)}+{{\bm G}}_1 {{\bm u}}, i=1,\ldots,I_1.
\end{align}
For notational convenience, we let
\rev{${\bm g}_{1}\triangleq{{\bm G}}_1 {{\bm u}}\in {\mathbb{C}^{N\times 1 }}$, ${\bar{\bm d}}\triangleq {\tilde{\bm u}}+\sum_{m=1}^{M_1}{\tilde{\bm d}}_{m} \in {\mathbb{C}^{M_2\times 1 }}$}, and 
\begin{align}\label{Q}
{\bar{\bm Q}}\triangleq{{\bm G}}_2 ~{\rm diag} \left({\bar{\bm d}}\right)={\tilde{\bm R}}+\sum_{m=1}^{M_1}  {{\bm Q}}_{m}
\end{align}
where \rev{${\bar{\bm Q}}\in {\mathbb{C}^{N\times M_2 }}$} can be regarded as the superimposed CSI of the single- and double-reflection channels related to IRS~2. 
As such, the effective channel in \eqref{superposed5} can be simplified as
\begin{align}\label{superposed5.0}
{\bm h}_{\rm I}^{(i)}={\bar{\bm Q}}{\bm \theta}_{2,\rm I}^{(i)}+{\bm g}_{1}, \quad i=1,\ldots,I_1.
\end{align}


Based on the channel model in \eqref{superposed5.0} with
$x_{\rm I}^{(i)}=1$ being the pilot symbol transmitted by the (single) user, the received signal at the BS during symbol~$i$ of Phase I is expressed as
\begin{align}
{\bm z}_{\rm I}^{(i)}= {\bar{\bm Q}}~ {\bm \theta}_{2,\rm I}^{(i)} +{\bm g}_{1} +{\bm v}_{\rm I}^{(i)}, \quad i=1,\ldots,I_1
\end{align}
where ${\bm v}_{\rm I}^{(i)}\sim {\mathcal N_c }({\bm 0}, \sigma^2{\bm I}_N)$ is the additive white Gaussian noise (AWGN) vector with normalized noise power of $\sigma^2$ at the BS (with respect to user transmit power). 
By stacking the received
signal vectors $\left\{{\bm z}_{\rm I}^{(i)}\right\}_{i=1}^{I_1}$ into \rev{${\bm Z}_{\rm I}=\left[{\bm z}_{\rm I}^{(1)},{\bm z}_{\rm I}^{(2)},\ldots,{\bm z}_{\rm I}^{(I_1)}\right]\in {\mathbb{C}^{N\times I_1 }}$}, we obtain
\begin{align}\label{rec_z}
{\bm Z}_{\rm I}= \left[{\bm g}_{1},{\bar{\bm Q}}\right] \underbrace{\begin{bmatrix}
	1,&1,&\ldots,&1\\{\bm \theta}_{2,\rm I}^{(1)},&{\bm \theta}_{2,\rm I}^{(2)},&\ldots,&{\bm \theta}_{2,\rm I}^{(I_1)}
	\end{bmatrix}}_{{\bar{\bm \Theta}}_{2,\rm I}} +{\bm V}_{\rm I}.
\end{align}
where
${\bar{\bm \Theta}}_{2,\rm I}\in {\mathbb{C}^{(M_2+1)\times I_1 }}$
denotes the training reflection matrix of IRS~2 in Phase~I and \rev{${\bm V}_{\rm I}=\left[{\bm v}_{\rm I}^{(1)},{\bm v}_{\rm I}^{(2)},\ldots,{\bm v}_{\rm I}^{(I_1)}\right]\in {\mathbb{C}^{N\times I_1 }}$} is the corresponding AWGN matrix.
By properly constructing the training
reflection matrix of IRS~2 in Phase~I such that ${\rm rank}\left( {\bar{\bm \Theta}}_{2,\rm I} \right)=M_2+1$, the least-square (LS) estimates of ${\bm g}_{1}$ and ${\bar{\bm Q}}$ based on \eqref{rec_z} are given by
\begin{align}\label{est_Q2}
\left[{\hat{\bm g}_{1}}, {\hat{\bar{\bm Q}}} \right]={\bm Z}_{\rm I} {\bar{\bm \Theta}}_{2,\rm I}^{\dagger}=\left[{\bm g}_{1},{\bar{\bm Q}}\right]+{\bm V}_{\rm I}{\bar{\bm \Theta}}_{2,\rm I}^{\dagger}
\end{align}
where ${\bar{\bm \Theta}}_{2,\rm I}^{\dagger}={\bar{\bm \Theta}}_{2,\rm I}^H\left({\bar{\bm \Theta}}_{2,\rm I} {\bar{\bm \Theta}}_{2,\rm I}^H\right)^{-1}$.
Note that $I_1\ge M_2+1$ is
required to ensure ${\rm rank}\left( {\bar{\bm \Theta}}_{2,\rm I} \right)=M_2+1$ and thus
the existence of ${\bar{\bm \Theta}}_{2,\rm I}^{\dagger}$.
\begin{lemma}
	With the superimposed CSI ${\bar{\bm Q}}$ available after Phase I, the channel model in \eqref{superposed3} for the single-user case can be simplified as
	\begin{align}\label{superposed6}
	{\bm h}=
	\left[ {\bar{\bm Q}} {\rm diag} \left({\bm \theta}_2\right) {{\bm E}},~{{\bm R}}\right]
	{\vec{\bm \theta}}_1
	\end{align}
	where \rev{${\vec{\bm \theta}}_1\triangleq\left[1,{\bm \theta}_1^T,{\bm \theta}_1^T\right]^T \in {\mathbb{C}^{(2M_1+1)\times 1 }}$} and ${{\bm E}}\triangleq\left[{{\bm e}}_{0}, {{\bm e}}_{1},\ldots, {{\bm e}}_{M_1}\right]\in {\mathbb{C}^{M_2\times (M_1+1) }}$ is the scaling matrix that collects all the (lower-dimensional) scaling vectors $\left\{{{\bm e}}_{m}\right\}_{m=0}^{M_1}$ with 
	${{\bm e}}_{0} \triangleq{\rm diag} \left({\bar{\bm d}}\right)^{-1}{\tilde{\bm u}}\in {\mathbb{C}^{M_2\times 1}}$ and
	${{\bm e}}_{m} \triangleq{\rm diag} \left({\bar{\bm d}}\right)^{-1}{\tilde{\bm d}}_{m}\in {\mathbb{C}^{M_2\times 1}}, \forall m=1,\ldots,M_1$.
\end{lemma}
\begin{proof}
Recall from \eqref{superposed3} that ${\tilde{\bm R}}$ and $\left\{{{\bm Q}}_{m} \right\}_{m=1}^{M_1}$ are respectively expressed as
\begin{align}\label{cascaded_Qk}
 {\tilde{\bm R}}={{\bm G}}_2   {\rm diag} \left( {\tilde{\bm u}} \right), \quad
{{\bm Q}}_{m}={{\bm G}}_2 ~{\rm diag} \left({\tilde{\bm d}}_{m}\right)
\end{align}
with $m=1,\ldots, M_1$.
It is observed that ${\tilde{\bm R}}$ and $\left\{ {{\bm Q}}_{m} \right\}_{m=1}^{M_1}$ in \eqref{cascaded_Qk} share the same (common) IRS~2$\rightarrow$BS link (i.e, ${{\bm G}}_2$) as the superimposed CSI ${\bar{\bm Q}}$ in \eqref{Q}.
As such,
if taking the superimposed CSI ${\bar{\bm Q}}$ in \eqref{Q} as the reference CSI, we can re-express \eqref{cascaded_Qk} as
\begin{align}
{\tilde{\bm R}}
&=\underbrace{ {{\bm G}}_2 ~{\rm diag} \left({\bar{\bm d}}\right) }_{{\bar{\bm Q}} } \cdot \underbrace{ {\rm diag} \left({\bar{\bm d}}\right)^{-1}{\rm diag} \left({\tilde{\bm u}}\right)}_{{\rm diag} \left({{\bm e}}_{0}\right)}\label{cascaded_Rk2}\\
{{\bm Q}}_{m}
&=\underbrace{ {{\bm G}}_2 ~{\rm diag} \left({\bar{\bm d}}\right) }_{{\bar{\bm Q}} } \cdot \underbrace{ {\rm diag} \left({\bar{\bm d}}\right)^{-1}{\rm diag} \left({\tilde{\bm d}}_{m}\right)}_{{\rm diag} \left({{\bm e}}_{m}\right)} \label{cascaded_Qk2}
\end{align}
with $m=1,\ldots, M_1$,
where ${{\bm e}}_{0} \triangleq{\rm diag} \left({\bar{\bm d}}\right)^{-1}{\tilde{\bm u}}\in {\mathbb{C}^{M_2\times 1}}$ and ${{\bm e}}_{m} \triangleq{\rm diag} \left({\bar{\bm d}}\right)^{-1}{\tilde{\bm d}}_{m}\in {\mathbb{C}^{M_2\times 1}}, \forall m=1,\ldots,M_1$ denote the scaling vectors normalized by ${\bar{\bm d}}$. As compared to the cascaded channel matrices ${\tilde{\bm R}}$ and $\left\{{{\bm Q}}_{m} \right\}_{m=1}^{M_1}$, the scaling vectors $\left\{{{\bm e}}_{m}\right\}_{m=0}^{M_1}$ are of lower dimension with ${\bar{\bm Q}}$ being the reference CSI.
By substituting ${\tilde{\bm R}}$ and ${{\bm Q}}_{m}$ of \eqref{cascaded_Rk2} and \eqref{cascaded_Qk2} into \eqref{superposed3}, the channel model in \eqref{superposed3} for the single user case can be equivalently expressed as
\begin{align}\label{superposed5.1}
{\bm h}&=\sum\limits_{m=1}^{M_1}  {\bar{\bm Q}} {\rm diag} \left({{\bm e}}_{m}\right)
 {\bm \theta}_2 {\theta}_{1,m}
+ {\bar{\bm Q}}{\rm diag} \left({{\bm e}}_{0}\right)
{\bm \theta}_2
+ {{\bm R}}{\bm \theta}_1\notag\\
&={\bar{\bm Q}} {\rm diag} \left({\bm \theta}_2\right) \underbrace{\left[{{\bm e}}_{0}, {{\bm e}}_{1},\ldots, {{\bm e}}_{M_1} \right]}_{{{\bm E}}}
\begin{bmatrix}
1\\{\bm \theta}_1
\end{bmatrix}
+{{\bm R}}{\bm \theta}_1.
\end{align}
\rev{Let ${\vec{\bm \theta}}_1\triangleq\left[1,{\bm \theta}_1^T,{\bm \theta}_1^T\right]^T\in {\mathbb{C}^{(2M_1+1)\times 1 }}$ for notational simplicity} and after some simple transformations,
the channel model in \eqref{superposed5.1} can be further expressed in a more compact form as in \eqref{superposed6}, thus completing the proof.
\end{proof}

According to Lemma 1, it is sufficient to acquire the CSI of $\left\{{\bar{\bm Q}}, {{\bm E}},{{\bm R}}\right\}$ for our considered cascaded channel estimation. 
Moreover, it is worth pointing out that as compared to \eqref{superposed3} with totally $N(M_1+M_2)+NM_1M_2$ channel coefficients in $\left\{\left\{ {{\bm Q}}_{m} \right\}_{m=1}^{M_1}, {\tilde{\bm R}}, {{\bm R}}\right\}$ under the single-user setup with $K=1$,
the number of channel coefficients in $\left\{{\bar{\bm Q}}, {{\bm E}},{{\bm R}}\right\}$ based on \eqref{superposed6} is substantially reduced to $N(M_1+M_2)+M_2(M_1+1)$ when $N\gg 1$ in practice.
Moreover, with the superimposed CSI ${\bar{\bm Q}}$ estimated in Phase I according to \eqref{est_Q2}, we only need to further estimate $\left\{{{\bm E}},{{\bm R}}\right\}$ in the subsequent Phase II, as elaborated in the next.

{\bf Phase II (Estimation of $\left\{{{\bm E}},{{\bm R}}\right\}$):} In this phase, we dynamically tune the reflection phase-shifts of the two IRSs, i.e.,
$\left\{{\bm \theta}_{1,\rm II}^{(i)}, {\bm \theta}_{2,\rm II}^{(i)}\right\}$, over different pilot symbols to facilitate the joint estimation of $\left\{{{\bm E}},{{\bm R}}\right\}$, where the number of pilot symbols in Phase II is denoted by $I_2$.
As such, the effective channel in \eqref{superposed6} during symbol~$i$ of Phase~II, denoted by ${\bm h}_{\rm II}^{(i)}$, is given by
\begin{align}\label{superposed6.0}
{\bm h}_{\rm II}^{(i)}=
\left[ {\bar{\bm Q}} {\rm diag} \left({\bm \theta}_{2,\rm II}^{(i)}\right) {{\bm E}},~{{\bm R}}\right]
{\vec{\bm \theta}}_{1,\rm II}^{(i)},  \quad i=1,\ldots,I_2
\end{align}
where ${\vec{\bm \theta}}_{1,\rm II}^{(i)}\triangleq\left[1,({\bm \theta}_{1,\rm II}^{(i)})^T,({\bm \theta}_{1,\rm II}^{(i)})^T\right]^T$.

Based on the channel model in \eqref{superposed6.0} with $x_{\rm II}^{(i)}=1$ being the pilot symbol transmitted by the (single) user, the received signal at the BS during symbol $i$ of Phase II is given by
\begin{align}\label{rec_z1}
{\bm z}_{\rm II}^{(i)}=
\left[ {\bar{\bm Q}} {\rm diag} \left({\bm \theta}_{2,\rm II}^{(i)}\right) {{\bm E}},~{{\bm R}}\right]
{\vec{\bm \theta}}_{1,\rm II}^{(i)}
+{\bm v}_{\rm II}^{(i)}
\end{align}
where ${\bm v}_{\rm II}^{(i)}\sim {\mathcal N_c }({\bm 0}, \sigma^2{\bm I}_N)$ is the AWGN vector at the BS. \rev{Note that ${\bar{\bm Q}}$ can be regarded as the spatial observation matrix for ${{\bm E}}$ in \eqref{rec_z1} and its column rank will affect the training reflection design and the corresponding minimum training overhead of Phase II.
As such, for the joint estimation of $\left\{{{\bm E}},{{\bm R}}\right\}$, we consider the following two cases depending on whether ${\bar{\bm Q}}$ is of full column rank or not.}
\subsubsection{Case 1}\label{case1} $N\ge M_2$. In this case, we design the time-varying reflection phase-shifts of IRS~2 as ${\bm \theta}_{2,\rm II}^{(i)}=\psi_{\rm II}^{(i)}{\bm \theta}_{2,\rm II}$, where $\psi_{\rm II}^{(i)}$ with $|\psi_{\rm II}^{(i)}|=1$ denotes the time-varying common phase shift that is applied to all the reflecting elements/subsurfaces of IRS~2 (given the initial reflection phase-shifts of IRS~2  in Phase II as ${\bm \theta}_{2,\rm II}$).
For simplicity, we set ${\bm \theta}_{2,\rm II}= {\bm 1}_{M_2\times 1}$ and the resultant received signal in \eqref{rec_z1} reduces to\footnote{Note that ${\bm \theta}_{2,\rm II}$ can be arbitrarily set in Phase II for our proposed channel estimation scheme in the case of $N\ge M_2$ without affecting its performance.}
\begin{align}\label{rec_z.1}
{\bm z}_{\rm II}^{(i)}&=
\left[ \psi_{\rm II}^{(i)} {\bar{\bm Q}} {{\bm E}},~{{\bm R}}\right]
{\vec{\bm \theta}}_{1,\rm II}^{(i)}
+{\bm v}_{\rm II}^{(i)}\notag\\
&=
\underbrace{\left[ {\bar{\bm Q}}{{\bm E}},~{{\bm R}}\right]}_{{{\bm F}}}
\begin{bmatrix}
\psi_{\rm II}^{(i)}\\\psi_{\rm II}^{(i)}{\bm \theta}_{1,\rm II}^{(i)}\\{\bm \theta}_{1,\rm II}^{(i)}
\end{bmatrix}
+{\bm v}_{\rm II}^{(i)}
\end{align}
where ${{\bm F}}\in {\mathbb{C}^{N\times (2M_1+1) }}$ denotes the composite CSI of $\left\{{\bar{\bm Q}},{{\bm E}},{{\bm R}}\right\}$.
By stacking the received
signal vectors $\left\{{\bm z}_{\rm II}^{(i)}\right\}$ over $I_2$ pilot symbols as \rev{${\bm Z}_{\rm II}=\left[{\bm z}_{\rm II}^{(1)},{\bm z}_{\rm II}^{(2)},\ldots,{\bm z}_{\rm II}^{(I_2)}\right]\in {\mathbb{C}^{N\times I_2 }}$}, we obtain
\begin{align}\label{rec_z2}
{\bm Z}_{\rm II}= {{\bm F}}  
\underbrace{\begin{bmatrix}
\psi_{\rm II}^{(1)},&\psi_{\rm II}^{(2)},&\ldots,&\psi_{\rm II}^{(I_2)}\\
\psi_{\rm II}^{(1)}{\bm \theta}_{1,\rm II}^{(1)},&\psi_{\rm II}^{(2)}{\bm \theta}_{1,\rm II}^{(2)},&\ldots,&\psi_{\rm II}^{(I_2)}{\bm \theta}_{1,\rm II}^{(I_2)}\\
{\bm \theta}_{1,\rm II}^{(1)},&{\bm \theta}_{1,\rm II}^{(2)},&\ldots,&{\bm \theta}_{1,\rm II}^{(I_2)}
\end{bmatrix}}_{{\bm \Omega}_{\rm II}}
+ {\bm V}_{\rm II} 
\end{align}
where ${\bm \Omega}_{\rm II}\in {\mathbb{C}^{(2M_1+1)\times I_2 }}$ denotes the joint training reflection matrix of
 the phase shifts $\left\{{\bm \theta}_{1,\rm II}^{(i)}\right\}_{i=1}^{I_2}$ at IRS~1 and the (common) phase shifts $\left\{\psi_{\rm II}^{(i)}\right\}_{i=1}^{I_2}$ at IRS~2, and \rev{${\bm V}_{\rm II}=\left[{\bm v}_{\rm II}^{(1)},{\bm v}_{\rm II}^{(2)},\ldots,{\bm v}_{\rm II}^{(I_2)}\right]\in {\mathbb{C}^{N\times I_2 }}$} is the corresponding AWGN matrix in Phase II.
By properly constructing the joint training reflection matrix ${\bm \Omega}_{\rm II}$ such that ${\rm rank}\left( {\bm \Omega}_{\rm II} \right)=2M_1+1$, the LS estimate of ${{\bm F}} $ based on \eqref{rec_z2} is given by
\begin{align}\label{est_F1}
{\hat{{\bm F}}} = {\bm Z}_{\rm II} {\bm \Omega}_{\rm II}^{\dagger}={{\bm F}}
+ {\bm V}_{\rm II} {\bm \Omega}_{\rm II}^{\dagger}
\end{align}
where ${\bm \Omega}_{\rm II}^{\dagger}={\bm \Omega}_{\rm II}^H\left({\bm \Omega}_{\rm II} {\bm \Omega}_{\rm II}^H\right)^{-1}$.
Note that $I_2\ge 2M_1+1$ is required to ensure ${\rm rank}\left( {\bm \Omega}_{\rm II} \right)=2M_1+1$ and thus the existence of ${\bm \Omega}_{\rm II}^{\dagger}$.
With the CSI of ${\bar{\bm Q}}$ and ${{\bm F}} $ estimated in \eqref{est_Q2} and \eqref{est_F1}, respectively, the LS estimates of ${{\bm E}}$ and ${{\bm R}}$ are respectively given by
\begin{align}
{\hat{{\bm E}}}&=\left( {\hat{\bar{\bm Q}}^H} {\hat{\bar{\bm Q}}}\right)^{-1}{\hat{\bar{\bm Q}}^H}\left[{\hat{{\bm F}}}\right]_{:,1:M_1+1} \label{est_Ek}\\
 {\hat{{\bm R}}}&=\left[{\hat{{\bm F}}}\right]_{:,M_1+2:2M_1+1}.\label{est_R}
\end{align}

With the CSI of $\left\{{\bar{\bm Q}},{{\bm E}},{{\bm R}}\right\}$ estimated in \eqref{est_Q2} and \eqref{est_F1}-\eqref{est_R} for the case of $N\ge M_2$, we can then recover the estimated CSI of the double-reflection channels $\left\{{{\bm Q}}_{m}\right\}_{m=1}^{M_1}$ and the single-reflection channel ${\tilde{\bm R}}$ according to \eqref{cascaded_Rk2} and \eqref{cascaded_Qk2} in the proof of Lemma 1.
As such, accounting for both Phases I and II, 
the minimum training overhead in terms of number of pilot symbols required is $2M_1+M_2+2$ for the proposed channel estimation scheme in the case of $N\ge M_2$ under the single-user setup. 

\subsubsection{Case 2} $N< M_2$. Recall from \eqref{rec_z.1} that $\left[{{\bm F}}\right]_{:,1:M_1+1}={\bar{\bm Q}}{{\bm E}}$.
In this case, since ${\bar{\bm Q}}$ is (column) rank-deficient, i.e., ${\rm rank}\left({\bar{\bm Q}}\right)=N< M_2$,
we cannot estimate ${{\bm E}}$ according to \eqref{est_Ek}.
As such, we propose to simultaneously tune the training reflection phase-shifts of the two IRSs over time to jointly estimate $\left\{{{\bm E}},{{\bm R}}\right\}$ for the case of $N< M_2$. 
Specifically, 
the received signal at the BS during symbol $i$ of Phase~II in \eqref{rec_z1} can be re-expressed as
\begin{align}\label{rec_z3}
{\bm z}_{\rm II}^{(i)}=& {\bar{\bm Q}} {\rm diag} \left({\bm \theta}_{2,\rm II}^{(i)}\right) {{\bm E}} {\tilde{\bm \theta}}_{1,\rm II}^{(i)}+{{\bm R}} {\bm \theta}_{1,\rm II}^{(i)}
+{\bm v}_{\rm II}^{(i)}\notag\\
\stackrel{(a)}{=}&\left(({\tilde{\bm \theta}}_{1,\rm II}^{(i)})^T\otimes {\bar{\bm Q}} {\rm diag} \left({\bm \theta}_{2,\rm II}^{(i)}\right) \right){\rm vec}({{\bm E}})\notag\\
&+\left(({\bm \theta}_{1,\rm II}^{(i)})^T\otimes {\bm I}_{N} \right) {\rm vec}({{\bm R}})+{\bm v}_{\rm II}^{(i)}
\notag\\
=&\left[({\tilde{\bm \theta}}_{1,\rm II}^{(i)})^T\otimes {\bar{\bm Q}} {\bm \Psi}_{\rm II}^{(i)} ,~
({\bm \theta}_{1,\rm II}^{(i)})^T\otimes {\bm I}_{N}  \right]
\begin{bmatrix}
{\rm vec}({{\bm E}}) \\{\rm vec}({{\bm R}})
\end{bmatrix}+{\bm v}_{\rm II}^{(i)}
\end{align}
where $(a)$ is obtained according to the property of Kronecker product for vectorization, ${\tilde{\bm \theta}}_{1,\rm II}^{(i)}\triangleq \big[1, ({\bm \theta}_{1,\rm II}^{(i)})^T\big]^T$, and ${\bm \Psi}_{\rm II}^{(i)}\triangleq{\rm diag} \left({\bm \theta}_{2,\rm II}^{(i)}\right)$. 
By stacking the received
signal vectors $\left\{{\bm z}_{\rm II}^{(i)}\right\}$ over $I_2$ pilot symbols  as \rev{${\bm z}_{\rm II}=\left[({\bm z}_{\rm II}^{(1)})^T,({\bm z}_{\rm II}^{(2)})^T,\ldots,({\bm z}_{\rm II}^{(I_2)})^T\right]^T\in {\mathbb{C}^{I_2 N  \times1}}$}, we obtain
\begin{align}\label{rec_z4}
\hspace{-0.2cm}{\bm z}_{\rm II}=\hspace{-0.15cm}
\underbrace{\begin{bmatrix}
	({\tilde{\bm \theta}}_{1,\rm II}^{(1)})^T\otimes {\bar{\bm Q}} {\bm \Psi}_{\rm II}^{(i)},
	&\left(({\bm \theta}_{1,\rm II}^{(1)})^T\otimes {\bm I}_{N} \right)\\
	\vdots&\vdots\\
	({\tilde{\bm \theta}}_{1,\rm II}^{(I_2)})^T\otimes {\bar{\bm Q}} {\bm \Psi}_{\rm II}^{(i)} ,
	&\left(({\bm \theta}_{1,\rm II}^{(I_2)})^T\otimes {\bm I}_{N} \right)
	\end{bmatrix}}_{{\bm \Xi}\in {\mathbb{C}^{ I_2N  \times (M_2+M_1M_2+NM_1) }}}\hspace{-0.2cm}
\begin{bmatrix}
{\rm vec}({{\bm E}}) \\{\rm vec}({{\bm R}})
\end{bmatrix}
+ {\bm v}_{\rm II}
\end{align}
where \rev{${\bm v}_{\rm II}=\left[({\bm v}_{\rm II}^{(1)})^T,({\bm v}_{\rm II}^{(2)})^T,\ldots,({\bm v}_{\rm II}^{(I_2)})^T\right]^T\in {\mathbb{C}^{I_2 N  \times1}}$} is the corresponding AWGN vector in Phase II.
Accordingly, if we design the training
reflection phase-shifts $\left\{{\bm \theta}_{1,\rm II}^{(i)},{\bm \theta}_{2,\rm II}^{(i)}\right\}_{i=1}^{I_2}$ at the two IRSs properly 
such that ${\rm rank}\left( {\bm \Xi} \right)=M_2+M_1M_2+NM_1$, the LS estimates of ${{\bm E}}$ and ${{\bm R}}$ based on \eqref{rec_z4} can be obtained as 
\begin{align}\label{est_Ek2}
\begin{bmatrix}
{\rm vec}({\hat{\bm E}}) \\{\rm vec}({\hat{\bm R}})
\end{bmatrix}
={\bm \Xi}^{\dagger}{\bm z}_{\rm II} =\begin{bmatrix}
{\rm vec}({{\bm E}}) \\{\rm vec}({{\bm R}})
\end{bmatrix} + {\bm \Xi}^{\dagger}{\bm v}_{\rm II}
\end{align}
where ${\bm \Xi}^{\dagger}=\left( {\bm \Xi}^H {\bm \Xi}\right)^{-1}{{\bm \Xi}^H}$.
Note that since $I_2 N \ge M_2+M_1M_2+NM_1 $ is required to ensure ${\rm rank}\left( {\bm \Xi} \right)=M_2+M_1M_2+NM_1$ with $I_2$ being an integer, we have $I_2 \ge \left\lceil\frac{M_2+M_1M_2+NM_1}{N}\right\rceil=\left\lceil\frac{(M_1+1)M_2}{N}\right\rceil+M_1$.

Similar to the case of $N\ge M_2$, after estimating the CSI of $\left\{{\bar{\bm Q}},{{\bm E}},{{\bm R}}\right\}$ based on \eqref{est_Q2} and \eqref{est_Ek2} for the case of $N< M_2$, we can also
obtain the estimated CSI of $\left\{{{\bm Q}}_{m}\right\}_{m=1}^{M_1}$ and ${\tilde{\bm R}}$ according to \eqref{cascaded_Rk2} and \eqref{cascaded_Qk2} in the proof of Lemma 1.
As such, 
the corresponding minimum training overhead in terms of number of pilot symbols is $\left\lceil\frac{(M_1+1)M_2}{N}\right\rceil+M_1+M_2+1$ over Phases I and II for the case of $N< M_2$ under the single-user setup.
Finally, it is worth pointing out that although the LS channel estimation based on \eqref{rec_z4}-\eqref{est_Ek2} for the case of $N< M_2$ can also be applied to the case of $N\ge M_2$, it is generally less efficient due to the more reflection phase-shifts of the two IRSs required during Phase~II as well as the higher complexity of the larger-size matrix inversion operation for ${\bm \Xi}^{\dagger}$, as compared to the case of $N\ge M_2$ in Section~\ref{case1}.

\subsection{Training Reflection Phase-Shift Design}\label{training_design}
In this subsection, we optimize the training reflection phase-shift design to minimize the channel estimation error over the two training phases proposed in Section \ref{ProposedCE}.

{\bf Phase I:}  The mean squared error (MSE) of the LS channel estimation in \eqref{est_Q2} is given by
\begin{align}\label{MSE_I}
\varepsilon_{\rm I}&=\frac{1}{N(M_2+1)} {\mathbb E}\left\{  \left\|\left[{\hat{\bm g}_{1}}, {\hat{\bar{\bm Q}}} \right]
-\left[{\bm g}_{1},{\bar{\bm Q}}\right]\right\|^{2}_F
\right\}\notag\\
&=\frac{1}{N(M_2+1)}  {\mathbb E}\left\{ \Big\| {\bm V}_{\rm I} {\bar{\bm \Theta}}_{2,\rm I}^{\dagger} \Big\|^{2}_F\right\}\notag\\
&=\frac{1}{N(M_2+1)}  \text{tr}\left\{ \left({\bar{\bm \Theta}}_{2,\rm I}^{\dagger} \right)^H 
{\mathbb E}\left\{ {\bm V}_{\rm I}^H  {\bm V}_{\rm I}  \right\}
{\bar{\bm \Theta}}_{2,\rm I}^{\dagger} 	\right\}\notag\\
&=\frac{\sigma^2}{M_2+1} 
\text{tr}\left\{ \left( {\bar{\bm \Theta}}_{2,\rm I}{\bar{\bm \Theta}}_{2,\rm I}^H \right)^{-1}	\right\}
\end{align}
where ${\mathbb E}\left\{ {\bm V}_{\rm I}^H  {\bm V}_{\rm I}  \right\} =\sigma^2 N{\bm I}_{I_1}$. 
Similar to \cite{zheng2019intelligent,zheng2020intelligent}, the MSE in \eqref{MSE_I} can be minimized if and only if ${\bar{\bm \Theta}}_{2,\rm I}{\bar{\bm \Theta}}_{2,\rm I}^H=I_1 {\bm I}_{M_2+1}$.
Accordingly, we can design the training reflection matrix ${\bar{\bm \Theta}}_{2,\rm I}$ of IRS~2 in Phase I as e.g., the submatrix of the $I_1 \times I_1$ 
discrete Fourier transform (DFT) matrix with its first $M_2+1$ rows and thus achieve the minimum MSE as $\varepsilon_{\rm I}^{\rm min}=\frac{\sigma^2}{I_1} $ for Phase~I.

{\bf Phase II:} Next, we jointly optimize the training reflection phase-shifts of the two IRSs in Phase~II to minimize the MSE, which is divided into the following two cases for $N\ge M_2$ and $N< M_2$, respectively, as in Section~\ref{ProposedCE}.

\subsubsection{Case 1} $N\ge M_2$. The MSE of the LS channel estimation in \eqref{est_F1} is given by
\begin{align}\label{MSE_II}
\varepsilon_{\rm II}&=\frac{1}{N(2M_1+1)} {\mathbb E}\left\{  \left\| {\hat{{\bm F}}} 
-{{\bm F}}\right\|^{2}_F
\right\}\notag\\
&=\frac{1}{N(2M_1+1)}  {\mathbb E}\left\{ \Big\| {\bm V}_{\rm II} {\bm \Omega}_{\rm II}^{\dagger} \Big\|^{2}_F\right\}\notag\\
&=\frac{1}{N(2M_1+1)}  \text{tr}\left\{ \left( {\bm \Omega}_{\rm II}^{\dagger} \right)^H 
{\mathbb E}\left\{ {\bm V}_{\rm II}^H  {\bm V}_{\rm II} \right\}
{\bm \Omega}_{\rm II}^{\dagger} 	\right\}\notag\\
&=\frac{\sigma^2}{2M_1+1} 
\text{tr}\left\{ \left( {\bm \Omega}_{\rm II} {\bm \Omega}_{\rm II}^H \right)^{-1}	\right\}
\end{align}
where ${\mathbb E}\left\{ {\bm V}_{\rm II}^H  {\bm V}_{\rm II}  \right\} =\sigma^2 N{\bm I}_{I_2}$.
Similarly, the MSE in \eqref{MSE_II} can be minimized if and only if ${\bm \Omega}_{\rm II} {\bm \Omega}_{\rm II}^H=I_2 {\bm I}_{2M_1+1}$. However, recall from \eqref{rec_z2} that
${\bm \Omega}_{\rm II}$ is the joint training reflection matrix of
the phase-shifts $\left\{{\bm \theta}_{1,\rm II}^{(i)}\right\}_{i=1}^{I_2}$ at IRS~1 and the (common) phase shifts $\left\{\psi_{\rm II}^{(i)}\right\}_{i=1}^{I_2}$ at IRS~2, whose design is more involved as compared to that of ${\bar{\bm \Theta}}_{2,\rm I}$ in Phase~I.
For the purpose of exposition, \rev{let ${\bm \psi}_{\rm II}^H=\left[\psi_{\rm II}^{(1)},\ldots,\psi_{\rm II}^{(I_2)}\right]\in {\mathbb{C}^{1  \times I_2}}$} denote the (common) phase-shift vector at IRS~2 and \rev{${\bm \Theta}_{1,\rm II}=\left[{\bm \theta}_{1,\rm II}^{(1)},\ldots,{\bm \theta}_{1,\rm II}^{(I_2)}\right] \in {\mathbb{C}^{M_1  \times I_2}}$} denote the reflection phase-shift matrix at IRS~1 in Phase~II. By substituting ${\bm \psi}_{\rm II}^H$ and ${\bm \Theta}_{1,\rm II}$ into the joint training reflection matrix ${\bm \Omega}_{\rm II}$ in \eqref{rec_z2}, we have
\begin{align}
\hspace{-0.15cm}&{\bm \Omega}_{\rm II} {\bm \Omega}_{\rm II}^H=
\begin{bmatrix}
{\bm \psi}_{\rm II}^H  \\{\bm \Theta}_{1,\rm II} {\rm diag}({\bm \psi}_{\rm II}^H)\\{\bm \Theta}_{1,\rm II} 
\end{bmatrix}
\left[{\bm \psi}_{\rm II},~{\rm diag}({\bm \psi}_{\rm II}){\bm \Theta}_{1,\rm II}^H,~{\bm \Theta}_{1,\rm II}^H\right]\notag\\
\hspace{-0.15cm}&\stackrel{(b)}{=}\hspace{-0.15cm}\begin{bmatrix}
I_2  &{\bm 1}_{1\times I_2} {\bm \Theta}_{1,\rm II}^H&{\bm \psi}_{\rm II}^H{\bm \Theta}_{1,\rm II}^H\\
{\bm \Theta}_{1,\rm II} {\bm 1}_{I_2\times 1}&{\bm \Theta}_{1,\rm II} {\bm \Theta}_{1,\rm II}^H&{\bm \Theta}_{1,\rm II} {\rm diag}({\bm \psi}_{\rm II}^H){\bm \Theta}_{1,\rm II}^H\\
{\bm \Theta}_{1,\rm II} {\bm \psi}_{\rm II}& {\bm \Theta}_{1,\rm II} {\rm diag}({\bm \psi}_{\rm II}){\bm \Theta}_{1,\rm II}^H&{\bm \Theta}_{1,\rm II} {\bm \Theta}_{1,\rm II}^H
\end{bmatrix}
\end{align}
where $(b)$ holds since ${\bm \psi}_{\rm II}^H{\bm \psi}_{\rm II}=I_2$, ${\rm diag}({\bm \psi}_{\rm II}^H){\bm \psi}_{\rm II}={\bm 1}_{I_2\times 1}$,
and ${\rm diag}({\bm \psi}_{\rm II}^H){\rm diag}({\bm \psi}_{\rm II})={\bm I}_{I_2}$ for the (common) phase-shift vector at IRS~2.
Then the optimal condition of ${\bm \Omega}_{\rm II} {\bm \Omega}_{\rm II}^H=I_2 {\bm I}_{2M_1+1}$ to achieve the minimum MSE in \eqref{MSE_II} is equivalent to the following conditions:
\begin{align}
{\bm \Theta}_{1,\rm II} {\bm \Theta}_{1,\rm II}^H&=I_2 {\bm I}_{M_1},\label{con1}\\
 {\bm \Theta}_{1,\rm II} {\bm 1}_{I_2\times 1}&={\bm 0}_{M_1\times 1},\label{con2}\\
{\bm \Theta}_{1,\rm II} {\bm \psi}_{\rm II}&={\bm 0}_{M_1\times 1},\label{con3}\\
{\bm \Theta}_{1,\rm II} {\rm diag}({\bm \psi}_{\rm II}^H){\bm \Theta}_{1,\rm II}^H&={\bm 0}_{M_1\times M_1}.\label{con4}
\end{align}
As such, we need to carefully construct ${\bm \Theta}_{1,\rm II} $ and ${\bm \psi}_{\rm II}^H$ to satisfy all the conditions 
in \eqref{con1}-\eqref{con4} under the full-reflection constraint with the always-ON IRSs so as to minimize the MSE in \eqref{MSE_II}, which, however, is not a trivial task. 

Fortunately, we notice that the DFT matrix has the perfect orthogonality desired and the summation of each row (except the first row with all-one elements) is $0$, which can be exploited for the joint training reflection design of ${\bm \Theta}_{1,\rm II} $ and ${\bm \psi}_{\rm II}$. 
In particular, we need to carefully select $M+1$ rows from the DFT matrix to construct
${\bm \Theta}_{1,\rm II} $ and ${\bm \psi}_{\rm II}$ 
for achieving perfect orthogonality in ${\bm \Omega}_{\rm II}$,
so as to
minimize the MSE in \eqref{MSE_II}.\footnote{The performance of the heuristic training reflection design of ${\bm \Theta}_{1,\rm II} $ and ${\bm \psi}_{\rm II}$ by drawing arbitrary $M_1+1$ rows from the DFT matrix will be evaluated by simulations in Section \ref{Sim}, which is shown to be much worse than that of our proposed training reflection design.}
To this end, we propose to construct an optimal joint training reflection design of ${\bm \Theta}_{1,\rm II} $ and ${\bm \psi}_{\rm II}$ as follows.
\begin{figure}[!t]
	\centering
	\includegraphics[width=3.5in]{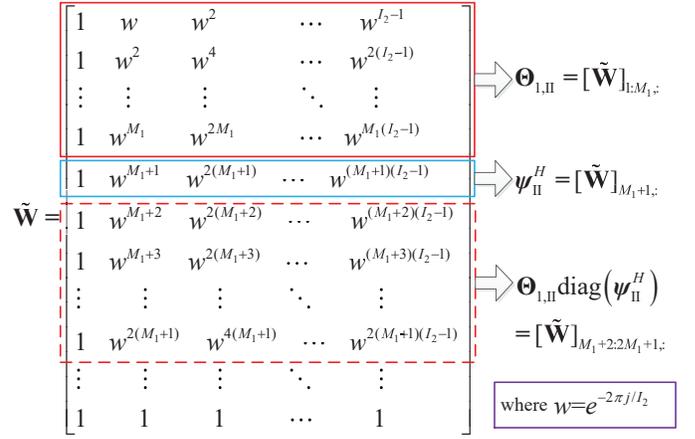}
	\setlength{\abovecaptionskip}{-6pt}
	\caption{Illustration of constructing an optimal joint training design of ${\bm \Theta}_{1,\rm II} $ and ${\bm \psi}_{\rm II}$.}
	\label{DFT matrix}
\end{figure}

Let ${\bm W}$ denote the $I_2 \times I_2$ DFT matrix,
where $I_2\ge 2M_1+1$ is required for the LS estimation in \eqref{est_F1}. First, we move the first 
row of ${\bm W}$ to the last and form a new $I_2 \times I_2$ matrix denoted by ${\tilde{\bm W}}$.
To satisfy the conditions in \eqref{con1}-\eqref{con2}, we let the training reflection matrix at IRS~1 
as the first $M_1$ rows of ${\tilde{\bm W}}$, i.e.,
${\bm \Theta}_{1,\rm II}=[{\tilde{\bm W}}]_{1:M_1,:}$.
Then it can be observed that, if we design the common phase-shift vector at IRS~2 as the $(M_1+1)$-th row of ${\tilde{\bm W}}$, i.e., ${\bm \psi}_{\rm II}^H=[{\tilde{\bm W}}]_{M_1+1,:}$, we then have ${\bm \Theta}_{1,\rm II} {\rm diag}({\bm \psi}_{\rm II}^H)=[{\tilde{\bm W}}]_{M_1+2:2M_1+1,:}$ due to the row-shifting effect of ${\rm diag}({\bm \psi}_{\rm II}^H)$ on ${\bm \Theta}_{1,\rm II}$ and thus the conditions in \eqref{con3}-\eqref{con4} can be simultaneously satisfied due to the pairwise orthogonality among rows in ${\tilde{\bm W}}$. The above optimal training reflection matrix construction is illustrated in Fig.~\ref{DFT matrix}.
With the optimal ${\bm \Theta}_{1,\rm II} $ and ${\bm \psi}_{\rm II}$ that satisfy ${\bm \Omega}_{\rm II} {\bm \Omega}_{\rm II}^H=I_2 {\bm I}_{2M_1+1}$, the minimum MSE of \eqref{MSE_II} is achieved, which is given by $\varepsilon_{\rm II}^{\rm min}=\frac{\sigma^2}{I_2}$ for Phase~II.


\subsubsection{Case 2} $N< M_2$.
The MSE of the LS channel estimation in \eqref{est_Ek2} is given by
\begin{align}\label{MSE_IIc2}
\varepsilon_{\rm II}'&=\frac{1}{M_2+M_1M_2+NM_1} {\mathbb E}\left\{  \left\| \begin{bmatrix}
{\rm vec}({\hat{\bm E}}) \\{\rm vec}({\hat{\bm R}})
\end{bmatrix}
-\begin{bmatrix}
{\rm vec}({{\bm E}}) \\{\rm vec}({{\bm R}})
\end{bmatrix}\right\|^{2}
\right\} \notag\\
&=\frac{1}{M_2+M_1M_2+NM_1}  {\mathbb E}\left\{ \Big\| {\bm \Xi}^{\dagger}{\bm v}_{\rm II} \Big\|^{2}\right\}\notag\\
&=\frac{1}{M_2+M_1M_2+NM_1}  \text{tr}\left\{ {\bm \Xi}^{\dagger} 	
{\mathbb E}\left\{ {\bm v}_{\rm II}  {\bm v}_{\rm II}^H \right\}
\left( {\bm \Xi}^{\dagger} \right)^H \right\}\notag\\
&=\frac{\sigma^2}{M_2+M_1M_2+NM_1} 
\text{tr}\left\{ \left( {\bm \Xi}^H {\bm \Xi} \right)^{-1}	\right\}
\end{align}
where ${\mathbb E}\left\{ {\bm v}_{\rm II}  {\bm v}_{\rm II}^H \right\}=\sigma^2 {\bm I}_{I_2 N}$.
Moreover, according to the expression of ${\bm \Xi}$ given in \eqref{rec_z4}, we have \eqref{matA} shown at the top of next page, where $(c)$ in \eqref{matA} is obtained according to the mixed-product property of Kronecker product.
\begin{figure*}
\begin{align}\label{matA}
{\bm \Xi}^H {\bm \Xi}=&
\begin{bmatrix}
({\tilde{\bm \theta}}_{1,\rm II}^{(1)})^*\otimes ( {\bar{\bm Q}}{\bm \Psi}_{\rm II}^{(1)})^H 
&\cdots
&({\tilde{\bm \theta}}_{1,\rm II}^{(I_2)})^*\otimes ( {\bar{\bm Q}}{\bm \Psi}_{\rm II}^{(I_2)})^H \\
({\bm \theta}_{1,\rm II}^{(1)})^*\otimes {\bm I}_{N} 
&\cdots
&({\bm \theta}_{1,\rm II}^{(I_2)})^*\otimes {\bm I}_{N} 
\end{bmatrix}
\begin{bmatrix}
({\tilde{\bm \theta}}_{1,\rm II}^{(1)})^T\otimes {\bar{\bm Q}}{\bm \Psi}_{\rm II}^{(1)} 
&({\bm \theta}_{1,\rm II}^{(1)})^T\otimes {\bm I}_{N} \\
\vdots&\vdots\\
({\tilde{\bm \theta}}_{1,\rm II}^{(I_2)})^T\otimes {\bar{\bm Q}}{\bm \Psi}_{\rm II}^{(I_2)} 
&({\bm \theta}_{1,\rm II}^{(I_2)})^T\otimes {\bm I}_{N} 
\end{bmatrix}\notag\\
\stackrel{(c)}{=}&\sum_{i=1}^{I_2}\begin{bmatrix}  
\left(({\tilde{\bm \theta}}_{1,\rm II}^{(i)})^*({\tilde{\bm \theta}}_{1,\rm II}^{(i)})^T\right)\otimes 
\left( ( {\bar{\bm Q}}{\bm \Psi}_{\rm II}^{(i)})^H {\bar{\bm Q}} {\bm \Psi}_{\rm II}^{(i)}\right)
&  
\left(({\tilde{\bm \theta}}_{1,\rm II}^{(i)})^*({\bm \theta}_{1,\rm II}^{(i)})^T\right)\otimes 
\left(({\bar{\bm Q}}{\bm \Psi}_{\rm II}^{(i)})^H \right)
\\
\left(({\bm \theta}_{1,\rm II}^{(i)})^*({\tilde{\bm \theta}}_{1,\rm II}^{(i)})^T\right)\otimes 
\left(  {\bar{\bm Q}} {\bm \Psi}_{\rm II}^{(i)} \right)
& 
\left(({\bm \theta}_{1,\rm II}^{(i)})^*({\bm \theta}_{1,\rm II}^{(i)})^T\right)\otimes 
{\bm I}_{N}
\end{bmatrix}\hspace{-0.15cm}
\end{align}
	\hrulefill
\vspace{-0.2cm}
\end{figure*}

It is noted that for any arbitrary $ {\bar{\bm Q}}$ involved in \eqref{matA}, it is generally difficult (if not impossible in some special cases) to jointly design the optimal training
reflection phase-shifts $\left\{{\bm \theta}_{1,\rm II}^{(i)},{\bm \theta}_{2,\rm II}^{(i)}\right\}_{i=1}^{I_2}$ at the two IRSs in Phase~II to minimize the MSE in \eqref{MSE_IIc2}.
Alternatively, we can construct the training
reflection phase-shifts $\left\{{\bm \theta}_{1,\rm II}^{(i)},{\bm \theta}_{2,\rm II}^{(i)}\right\}_{i=1}^{I_2}$ based on some orthogonal matrices (e.g., the DFT matrix, Hadamard matrix, and circulant matrix generated by Zadoff-Chu sequence \cite{Polyphase1972Polyphase}) to achieve ${\rm rank}\left( {\bm \Xi} \right)=M_2+M_1M_2+NM_1$ required by the LS estimation in \eqref{est_Ek2}.
\section{Double-IRS Channel Estimation for Multi-User Case}\label{MU_design_ON}
For the multi-user channel estimation, a straightforward method is by adopting the single-user channel estimation in 
Section~\ref{SU_design_ON} to estimate the cascaded channels of $K$ users separately over consecutive time, which, however, increases the total training overhead by $K$ times as compared to the single-user case and thus is practically prohibitive if $K$ is large. 
To achieve more efficient channel training, we extend the proposed channel estimation scheme and training reflection design for the single-user case in Section \ref{SU_design_ON} to the general multi-user case in this section in a different way. The key idea is to exploit 
the fact that given the cascaded channel of an arbitrary user (referred to as the reference user) estimated as in the single-user case (cf. Section~\ref{SU_design_ON}), the other users' cascaded channels can be expressed as lower-dimensional scaled versions of it and thus estimated with substantially reduced channel training.
The detailed channel estimation scheme and training reflection design for the general multi-user case are elaborated in the two subsequent subsections, respectively.

\subsection{Extended Channel Estimation Scheme for Multiple Users}
After estimating the cascaded channel of an arbitrary user as in the single-user case of Section~\ref{SU_design_ON}, the cascaded channels of the other users can be efficiently estimated by exploiting the property that all the users
share the common IRS~2$\rightarrow$BS (i.e, ${{\bm G}}_2$), IRS~1$\rightarrow$BS (i.e, ${{\bm G}}_1$), and IRS~1$\rightarrow$IRS~2 (i.e, ${\bm D}$) links in their respective single- and double-reflection channels (see Fig.~\ref{system}). In particular, we have the following lemma for the cascaded channels of the other users.
\begin{lemma}
With the cascaded CSI of an arbitrary user (say, $\left\{{{\bm Q}}_{1,m}\right\}_{m=1}^{M_1}$, ${\tilde{\bm R}}_{1}$, and ${{\bm R}}_{1}$ of user 1) available as the reference CSI,
we only need to estimate the scaling vectors \rev{${{\bm b}}_{k}={\rm diag} \left({{\bm u}}_{1}\right)^{-1} {{\bm u}}_{k} \in {\mathbb{C}^{M_1\times 1}}$ and
${\tilde {\bm b}}_{k}={\rm diag} \left({\tilde{\bm u}}_{1}\right)^{-1} {\tilde{\bm u}}_{k} \in {\mathbb{C}^{M_2\times 1}}$ }
to acquire the cascaded CSI of the remaining $K-1$ users with $k=2,\ldots,K$.
\end{lemma}
\begin{proof}
If given the cascaded CSI of user 1 as the reference CSI,
we can rewrite the two single-reflection channels $\left\{{{\bm R}}_{k}, {\tilde{\bm R}}_{k}\right\},\forall k=2,\ldots,K$ in \eqref{superposed3} as
\begin{align}
{{\bm R}}_{k}&={{\bm G}}_1   {\rm diag} \left( {{\bm u}}_{k} \right)=
{{\bm G}}_1{\rm diag} \left({{\bm u}}_{1}\right) \cdot {\rm diag}\left({{\bm u}}_{1}\right)^{-1} {{\bm u}}_{k}\notag\\
&=
{{\bm R}}_{1}{\rm diag} \left({{\bm b}}_{k} \right)\label{R_k1}\\
{\tilde{\bm R}}_{k}&={{\bm G}}_2   {\rm diag} \left( {\tilde{\bm u}}_{k} \right)=
{{\bm G}}_2{\rm diag} \left({\tilde{\bm u}}_{1}\right) \cdot {\rm diag}\left({\tilde{\bm u}}_{1}\right)^{-1} {\tilde{\bm u}}_{k}\notag\\
&=
{\tilde{\bm R}}_{1} {\rm diag} \left( {\tilde {\bm b}}_{k} \right)\label{R_k2}
\end{align}
and the double-reflection channels $\left\{{{\bm Q}}_{k,m}\right\}_{m=1}^{M_1},\forall k=2,\ldots,K$ in \eqref{superposed3} as
\begin{align}\label{cascaded_Q4}
{{\bm Q}}_{k,m}=&{{\bm G}}_2 {\rm diag} \left({\tilde{\bm d}}_{k,m}\right)=
{{\bm G}}_2 {\rm diag} \left({\bm d}_{m} {{u}}_{k,m} \right)\notag\\
=&{{\bm G}}_2 {\rm diag} \left({\bm d}_{m} {{u}}_{1,m}\right) \cdot{{u}}_{1,m}^{-1} {{u}}_{k,m}= {{\bm Q}}_{1,m}{{b}}_{k,m}
\end{align}
with $m=1,\ldots,M_1$,
\rev{where ${{\bm b}}_{k}\triangleq\left[{{b}}_{k,1},\ldots,{{b}}_{k,M_1}\right]^T={\rm diag} \left({{\bm u}}_{1}\right)^{-1} {{\bm u}}_{k}\in {\mathbb{C}^{M_1\times 1}}$ and
${\tilde {\bm b}}_{k}={\rm diag} \left({\tilde{\bm u}}_{1}\right)^{-1} {\tilde{\bm u}}_{k}\in {\mathbb{C}^{M_2\times 1}}$} are two scaling vectors
of user~$k$$\rightarrow$IRS~1 and user~$k$$\rightarrow$IRS~2 channels normalized by ${{\bm u}}_{1}$ and ${\tilde{\bm u}}_{1}$, respectively.
As such, with the cascaded CSI of user 1 available as the reference CSI, we only need to further estimate $\left\{{{\bm b}}_{k},{\tilde {\bm b}}_{k}\right\}_{k=2}^{K}$ to acquire the cascaded CSI of the remaining $K-1$ users according to the channel relationship in \eqref{R_k1}-\eqref{cascaded_Q4}, thus completing the proof.
\end{proof}
According to Lemma 2, after acquiring the cascaded CSI of user 1 as for the single-user case in
Section~\ref{SU_design_ON}, we only need to further estimate the scaling vectors $\left\{{{\bm b}}_{k},{\tilde {\bm b}}_{k}\right\}_{k=2}^{K}$ for the remaining $K-1$ users.
By substituting \eqref{R_k1}-\eqref{cascaded_Q4} into \eqref{superposed3}, we can re-express the channel model in \eqref{superposed3} as 
\begin{align}
&{\bm h}_k
=\sum_{m=1}^{M_1} {{\bm Q}}_{1,m} {\bm \theta}_2 {\theta}_{1,m} {{b}}_{k,m}
+{\tilde{\bm R}}_{1} {\rm diag} \left( {\tilde {\bm b}}_{k} \right){\bm \theta}_2
+{{\bm R}}_{1}{\rm diag} \left({{\bm b}}_{k} \right){\bm \theta}_1\notag\\
&=\left[{{\bm Q}}_{1,1} {\bm \theta}_2,\ldots, {{\bm Q}}_{1,M} {\bm \theta}_2\right] {\rm diag} \left({\bm \theta}_1\right)
{{\bm b}}_{k} \notag\\
&~~~+{\tilde{\bm R}}_{1} {\rm diag} \left({\bm \theta}_2\right) {\tilde {\bm b}}_{k}
+{{\bm R}}_{1} {\rm diag} \left({\bm \theta}_1\right) {{\bm b}}_{k}\notag\\
&=\underbrace{  \left[  \left(\left[{{\bm Q}}_{1,1} {\bm \theta}_2,\ldots, {{\bm Q}}_{1,M} {\bm \theta}_2\right]+{{\bm R}}_{1}\right){\bm \Phi}_1 ,~ {\tilde{\bm R}}_{1}{\bm \Phi}_2\right] }_{{\bm B}\in {\mathbb{C}^{N\times (M_1+M_2) }}}
\underbrace{\begin{bmatrix}
	{{\bm b}}_{k}\\{\tilde {\bm b}}_{k}
	\end{bmatrix}}_{{\bm \lambda}_{k}}
\label{superposed7}
\end{align}
\rev{where ${\bm \lambda}_{k} \in {\mathbb{C}^{(M_1+M_2)\times 1}} $ denotes the the stacked scaling vector of user $k$.}
Following Phases I-II for estimating the cascaded CSI of the
single user (i.e., user 1) in Section~\ref{SU_design_ON},
we propose the joint estimation of $\left\{ {\bm \lambda}_{k} \right\}_{k=2}^{K}$ 
in the subsequent Phase III, which is elaborated in the next.

{\bf Phase III (Estimation of $\{{{\bm b}}_{k},{\tilde {\bm b}}_{k}\}_{k=2}^{K}$):} 
In this phase, we fix the reflection phase-shifts of the two IRSs for simplicity\footnote{Note that since the cascaded CSI of $\left\{{{\bm Q}}_{1,m}\right\}_{m=1}^{M_1}$, ${\tilde{\bm R}}_{1}$, and ${{\bm R}}_{1}$ has been acquired in Phases~I-II, we can next optimize the IRS reflection phase-shifts $\left\{{\bm \theta}_{1,\rm III}, {\bm \theta}_{2,\rm III}\right\}$ in ${\bm B}$ to facilitate the joint estimation of $\left\{ {\bm \lambda}_{k}\right\}_{k=2}^{K}$ in Phase~III.}, i.e., ${\bm \theta}_{1,\rm III}^{(i)}={\bm \theta}_{1,\rm III}$ and ${\bm \theta}_{2,\rm III}^{(i)}={\bm \theta}_{2,\rm III}$, and 
design the concurrent pilot symbols for the remaining $K-1$ users to facilitate the
joint estimation of $\left\{ {\bm \lambda}_{k} \right\}_{k=2}^{K}$.
Based on the channel model in \eqref{superposed7} with the fixed reflection phase-shifts $\left\{{\bm \theta}_{1,\rm III}, {\bm \theta}_{2,\rm III}\right\}$ of the two IRSs, the received signal at the BS during symbol period~$i$ of Phase III can be expressed as
\begin{align}\label{rec_MU_ON}
\hspace{-0.2cm}{\bm z}_{\rm III}^{(i)}\hspace{-0.1cm}=\hspace{-0.1cm} \sum_{k=2}^{K} x_k^{(i)} {\bm B}  {\bm \lambda}_{k}
+{\bm v}_{\rm III}^{(i)}={\bm B} \underbrace{\left[{\bm \lambda}_{2},\ldots,{\bm \lambda}_{K}\right]}_{{\bm \Lambda}}
	 {\bm x}^{(i)}+{\bm v}_{\rm III}^{(i)}\hspace{-0.2cm}
\end{align}
where $i=1,\ldots,I_3$ with $I_3$ denoting the number of (pilot) symbol periods in Phase III, \rev{${\bm x}^{(i)}=\left[x_{2}^{(i)},\ldots,x_{K}^{(i)}\right]^T \in {\mathbb{C}^{(K-1)\times 1}}$} denotes the pilot symbol vector transmitted by the remaining $K-1$ users, and ${\bm v}_{\rm III}^{(i)}\sim {\mathcal N_c }({\bm 0}, \sigma^2{\bm I}_N)$ is the AWGN vector with normalized noise power of $\sigma^2$.
\rev{Note that ${\bm B}$ can be regarded as the spatial observation matrix for \eqref{rec_MU_ON} and its column rank will affect the training design and the corresponding minimum training overhead of Phase III.
As such, for the joint estimation of ${\bm \Lambda}\triangleq\left[{\bm \lambda}_{2},\ldots,{\bm \lambda}_{K}\right]\in {\mathbb{C}^{(M_1+M_2)\times (K-1)}}$, we consider the following two cases depending on whether ${\bm B}$ is of full column rank or not.}

\subsubsection{Case 1} $N\ge M_1+M_2$. 
In this case, we stack the received
signal vectors \rev{$\left\{{\bm z}_{\rm III}^{(i)}\right\}$} over $I_3$ symbol periods as \rev{${\bm Z}_{\rm III}=\left[{\bm z}_{\rm III}^{(1)},{\bm z}_{\rm III}^{(2)},\ldots,{\bm z}_{\rm III}^{(I_3)}\right]\in {\mathbb{C}^{N\times I_3}}$} and thus obtain
\begin{align}\label{rec_MU_ON2}
{\bm Z}_{\rm III}&=   {\bm B}  {\bm \Lambda} {\bm X}
+{\bm V}_{\rm III}
\end{align}
where 
\rev{${\bm X}=\left[{\bm x}^{(1)},{\bm x}^{(2)},\ldots,{\bm x}^{(I_3)}\right]\in {\mathbb{C}^{(K-1)\times I_3}}$} is the pilot symbol matrix of the remaining $K-1$ users and
\rev{${\bm V}_{\rm III}=\left[{\bm v}_{\rm III}^{(1)},{\bm v}_{\rm III}^{(2)},\ldots,{\bm v}_{\rm III}^{(I_3)}\right]\in {\mathbb{C}^{N\times I_3}}$} is the corresponding AWGN matrix.
By properly designing the training
reflection phase-shifts $\left\{{\bm \theta}_{1,\rm III}, {\bm \theta}_{2,\rm III}\right\}$ of the two IRSs as well as the pilot symbol matrix ${\bm X}$ such that ${\rm rank}\left(  {\bm B} \right)=M_1+M_2$ and ${\rm rank}\left( {\bm X}\right)=K-1$,
the LS estimate of ${\bm \Lambda}$ based on \eqref{rec_MU_ON2} is given by
\begin{align}\label{Est_b}
{\hat{\bm \Lambda}}= {\bm B}^{\dagger}{\bm Z}_{\rm III} {\bm X}^{\dagger} ={\bm \Lambda} +{\bm B}^{\dagger}{\bm V}_{\rm III} {\bm X}^{\dagger}
\end{align}
where ${\bm B}^{\dagger}=\left({\bm B}^H{\bm B}\right)^{-1}{\bm B}^H$ and ${\bm X}^{\dagger}={\bm X}^H \left({\bm X}{\bm X}^H\right)^{-1}$.
Note that $I_3\ge K-1$ is
required to ensure ${\rm rank}\left( {\bm X}\right)=K-1$ and thus the existence of ${\bm X}^{\dagger}$ for the case of $N\ge M_1+M_2$.

\subsubsection{Case 2} $N< M_1+M_2$. In this case, since ${\bm B}$ is (column) rank-deficient, i.e., ${\rm rank}\left({\bm B}\right)=N<  M_1+M_2$,
we cannot estimate ${\bm \Lambda}$ according to \eqref{Est_b}. 
Alternatively, we stack the received signal vectors $\{ {\bm z}_{\rm III}^{(i)}\}$ over $I_3$ symbol periods as \rev{${\bm z}_{\rm III}=\left[({\bm z}_{\rm III}^{(1)})^T,({\bm z}_{\rm III}^{(2)})^T,\ldots,({\bm z}_{\rm III}^{(I_3)})^T\right]^T={\rm vec} \left({\bm Z}_{\rm III}\right)\in {\mathbb{C}^{I_3 N\times 1}}$} and thus obtain
\begin{align}\label{rec_MU_ON3}
{\bm z}_{\rm III} 
= \left({\bm X}^T \otimes {\bm B}\right)  {\rm vec} \left({\bm \Lambda} \right)
+{\bm v}_{\rm III}
\end{align}
where \rev{${\bm v}_{\rm III}=\left[({\bm v}_{\rm III}^{(1)})^T,({\bm v}_{\rm III}^{(2)})^T,\ldots,({\bm v}_{\rm III}^{(I_3)})^T\right]^T\in {\mathbb{C}^{I_3 N\times 1}}$} is the corresponding AWGN vector in Phase III.
Accordingly, if we design the training
reflection phase-shifts $\left\{{\bm \theta}_{1,\rm III}, {\bm \theta}_{2,\rm III}\right\}$ of the two IRSs as well as the pilot symbol matrix ${\bm X}$ properly 
such that ${\rm rank}\left( {\bm X}^T \otimes {\bm B} \right)=(K-1)(M_1+M_2)$, the LS estimate of ${\rm vec} \left({\bm \Lambda} \right)$ based on \eqref{rec_MU_ON3} is given by
\begin{align}\label{est_bk}
\hspace{-0.2cm}{\rm vec} \left({\hat{\bm \Lambda}} \right)
\hspace{-0.1cm}=\hspace{-0.1cm}\left({\bm X}^T \otimes {\bm B}\right)^{\dagger} {\bm z}_{\rm III}  \hspace{-0.1cm}=\hspace{-0.1cm}{\rm vec} \left({\bm \Lambda} \right)
+\left({\bm X}^T \otimes {\bm B}\right)^{\dagger} {\bm v}_{\rm III}\hspace{-0.15cm}
\end{align}
where 
$\left({\bm X}^T \hspace{-0.05cm}\otimes\hspace{-0.05cm} {\bm B}\right)^{\dagger}\hspace{-0.1cm}=\hspace{-0.1cm}\left(\hspace{-0.1cm} \left({\bm X}^T \hspace{-0.05cm}\otimes\hspace{-0.05cm} {\bm B}\right)^H \left({\bm X}^T \hspace{-0.05cm}\otimes\hspace{-0.05cm} {\bm B}\right)\hspace{-0.1cm}\right)^{-1}
 \hspace{-0.1cm}\left({\bm X}^T \otimes {\bm B}\right)^H$.
Note that since $I_3 N \ge (K-1)(M_1+M_2) $ is required to ensure ${\rm rank}\left( {\bm X}^T \otimes {\bm B} \right)=(K-1)(M_1+M_2)$ with $I_3$ being an integer, we have $I_3 \ge \left\lceil\frac{(K-1)(M_1+M_2)}{N}\right\rceil$ for the case of $N< M_1+M_2$.
Moreover, it is worth pointing out that the LS channel estimation based on \eqref{est_bk} for the case of $N< M_1+M_2$ can also be applied to the case of $N\ge M_1+M_2$, but it incurs higher complexity due to the larger-size matrix inversion operation for $\left({\bm X}^T \otimes {\bm B}\right)^{\dagger}$.

\subsection{Training Design for Multiple Users}
It can be verified that for the case of $N \ge M_1+M_2$, the LS channel estimation based on \eqref{est_bk} 
is also equivalent to that based on \eqref{Est_b} with proper vectorization. 
As such, we focus on the MSE derivation for the LS channel estimation based on \eqref{est_bk}, which is given by
\begin{align}\label{MSE_III}
&\varepsilon_{\rm III}=\frac{1}{(K-1)(M_1+M_2)} {\mathbb E}\left\{  \left\|
{\rm vec} \left({\hat{\bm \Lambda}} \right) 
-{\rm vec} \left({\bm \Lambda} \right)\right\|^{2}
\right\} \notag\\
&=\frac{1}{(K-1)(M_1+M_2)}  {\mathbb E}\left\{ \Big\|\left({\bm X}^T \otimes {\bm B}\right)^{\dagger} {\bm v}_{\rm III} \Big\|^{2}\right\}\notag\\
&=\frac{1}{(K-1)(M_1+M_2)}  \text{tr}\hspace{-0.1cm}\left\{\hspace{-0.1cm} \left({\bm X}^T \hspace{-0.1cm}\otimes \hspace{-0.1cm}{\bm B}\right)^{\dagger}\hspace{-0.1cm}
{\mathbb E}\left\{ {\bm v}_{\rm III}  {\bm v}_{\rm III}^H \right\}
\hspace{-0.1cm}\left(\hspace{-0.05cm} \left({\bm X}^T \hspace{-0.1cm}\otimes\hspace{-0.1cm} {\bm B}\right)^{\dagger} \right)^H \hspace{-0.1cm}\right\}\notag\\
&=\frac{\sigma^2}{(K-1)(M_1+M_2)} 
\text{tr}\left\{ \left( \left({\bm X}^T \otimes {\bm B}\right)^H \left({\bm X}^T \otimes {\bm B}\right) \right)^{-1}	\right\}\notag\\
&\stackrel{(d)}{=}\frac{\sigma^2}{(K-1)(M_1+M_2)}  \text{tr}\left\{\hspace{-0.1cm} \left({\bm X}{\bm X}^H\right)^{-1}	\hspace{-0.1cm}\right\}
\text{tr}\left\{\hspace{-0.1cm} \left({\bm B}^H{\bm B}\right)^{-1}	\hspace{-0.1cm}\right\}\hspace{-0.2cm}
\end{align}
where ${\mathbb E}\left\{ {\bm v}_{\rm III}  {\bm v}_{\rm III}^H \right\}=\sigma^2 {\bm I}_{I_3 N}$ and $(d)$ is obtained according to the mixed-product and trace properties of Kronecker product.
As such, the MSE minimization in \eqref{MSE_III} is equivalent to minimizing $\text{tr}\left\{ \left({\bm X}{\bm X}^H\right)^{-1}	\right\}$ and $\text{tr}\left\{ \left({\bm B}^H{\bm B}\right)^{-1}	\right\}$, respectively.
Specifically, the minimization of $\text{tr}\left\{ \left({\bm X}{\bm X}^H\right)^{-1}	\right\}$ can be achieved if and only if 
${\bm X}{\bm X}^H=I_3{\bm I}_{K-1}$, which implies that the pilot sequences from the remaining $K-1$ users should be orthogonal to each other in Phase III. Accordingly, we can design the pilot symbol matrix ${\bm X}$ in Phase III as e.g., the submartix of the $I_3 \times I_3$ DFT matrix with its first $K-1$ rows.
On the other hand, since the matrix ${\bm B}$ involves the cascaded CSI of $\left\{{{\bm Q}}_{1,m}\right\}_{m=1}^{M_1}$, ${\tilde{\bm R}}_{1}$, and ${{\bm R}}_{1}$ which can be arbitrary in practice, it is generally difficult to design the optimal training reflection phase-shifts of the two IRSs to minimize $\text{tr}\left\{ \left({\bm B}^H{\bm B}\right)^{-1}	\right\}$. Alternatively, we can design the training reflection phase-shifts of the two IRSs based on some orthogonal matrices/sequences (e.g., the DFT matrix), to 
satisfy the rank conditions required by the LS channel estimations in \eqref{Est_b} and \eqref{est_bk}, whose performance will be examined by simulations in Section~\ref{Sim}.

\section{Numerical Results}\label{Sim}
In this section, we present numerical results to validate the effectiveness of our proposed channel estimation
scheme as well as the corresponding training design for the double-IRS aided multi-user MIMO system.
Under the three-dimensional (3D) Cartesian coordinate system, we assume that the central (reference) points of the BS, IRS~2, IRS~1, and user cluster are located at $(1,0,2)$, $(0,0.5,1)$, $(0,49.5,1)$, and $(1,50,0)$ in meter (m), respectively. 
Moreover, we assume that
the BS is equipped with a uniform linear array (ULA); while the two distributed IRSs are equipped with uniform planar arrays (UPAs).
As the element-grouping strategy adopted in \cite{yang2019intelligent,zheng2019intelligent},
each IRS subsurface is a small-size UPA composed of $5\times 5$ adjacent reflecting elements that share a common phase shift for reducing design complexity.
The distance-dependent channel path loss is modeled as $\gamma=\gamma_0/ d^\alpha$, where $\gamma_0$ denotes the path loss at the reference distance of 1~m which is set as $\gamma_0=-30$~dB for all individual links, $d$ denotes the individual link distance, and $\alpha$ denotes the path loss exponent which is set as $2.2$ for the link between the user cluster/BS and its nearby serving IRS (due to the short distance) and set as $3$ for the other links (due to the relatively longer distances). Without loss of generality, all the users are assumed to have equal transmit power, i.e., $P_k = P, \forall k=1,\ldots,K$ and the noise power at the BS is set as $\sigma^2_N=-65$ dBm.
Accordingly, the normalized noise power at the BS is given by $\sigma^2=\sigma^2_N/P$.

Note that there has been very limited work on channel estimation for the double-IRS aided system with the co-existence of single- and double-reflection links.
As such, we consider the following two benchmark channel estimation schemes for comparison.
\begin{itemize}
	\item {\bf Decoupled scheme with ON/OFF IRSs \cite{zheng2020Uplink}:} For the single-user case of this scheme, the cascaded channels of the two single-reflection links, each corresponding to one of the two IRSs respectively, are successively estimated at the multi-antenna BS with the other IRS turned OFF. Then, after canceling the signals over the two single-reflection channels estimated,
	the higher-dimensional double-reflection (i.e., user$\rightarrow$IRS~1$\rightarrow$IRS~2$\rightarrow$BS) channel is efficiently estimated at the BS by exploiting the fact that its cascaded channel coefficients are scaled
	versions of those of the single-reflection (i.e., user$\rightarrow$IRS~2$\rightarrow$BS) channel due to their commonly shared IRS~2$\rightarrow$BS link. For the multi-user case, by leveraging the same (common) channel relationship among users in Lemma 2, the scaling vectors $\left\{{{\bm b}}_{k}\right\}_{k=2}^{K}$ and $\left\{{\tilde {\bm b}}_{k}\right\}_{k=2}^{K}$ are separately estimated at the BS with one of the two IRSs turned OFF.
	\item {\bf Benchmark scheme based on \cite{you2020wireless}:} We extend the channel estimation method proposed in \cite{you2020wireless} as another benchmark scheme, where the double-reflection channel is estimated at each BS antenna in parallel without exploiting the (common) channel relationship with the single-reflection channels, and the cascaded channels of $K$ users are separately estimated over consecutive time. Moreover, as the single-reflection channels were ignored in \cite{you2020wireless}, the same decoupled channel estimation for the two single-reflection channels in \cite{zheng2020Uplink} is adopted for each user in this benchmark scheme.
\end{itemize}
\subsection{Training Overhead Comparison}\label{Training}
The training overhead comparison between the proposed channel estimation scheme and the two benchmark schemes is shown in Table~\ref{Table of estimation}, where $M_1=M_2=M/2$ is assumed for ease of exposition. As can be seen in Table~\ref{Table of estimation}, by exploiting the peculiar channel relationship over single- and double-reflection channels as well as that among multiple users, both the proposed channel estimation scheme with always-ON IRSs and the decoupled channel estimation scheme with ON/OFF IRSs incur much lower-order training overhead than the benchmark scheme based on \cite{you2020wireless}, especially when $N\ge M$. In the rest of this section, given the total number of subsurfaces $M=40$, we set $M_1=M_2=M/2=20$ for the two distributed IRSs.

 \begin{table*}[!t]
	\begin{center}\caption{Training Overhead Comparison for Different Channel Estimation Schemes (Assume $M_1=M_2=M/2$)}\label{Table of estimation}
		\resizebox{0.8\textwidth}{!}{
			\begin{tabular}{|c|c|c|c|}
				\hline
				\multirow{2}{*}{}                                & \multicolumn{3}{c|}{Minimum number of (pilot) symbol periods} \\ \cline{2-4} 
				& $N\ge M$        & $M>N\ge M/2$       & $N< M/2$       \\ \hline
				Proposed scheme with always-ON IRSs &
				$\frac{3}{2}M+K+1$ &
				$\frac{3}{2}M+2+\left\lceil\frac{(K-1)M}{N}\right\rceil$ &
				$M+1+\left\lceil\frac{(M+2)M}{4N}\right\rceil+\left\lceil\frac{(K-1)M}{N}\right\rceil$ \\ \hline
				Decoupled scheme with ON/OFF IRSs \cite{zheng2020Uplink} &
				\multicolumn{2}{c|}{$\frac{3}{2}M+2(K-1)$} &
				$M+\left\lceil\frac{M^2}{4N}\right\rceil+2\left\lceil\frac{(K-1) M}{2N}\right\rceil$ \\ \hline
				Benchmark scheme based on \cite{you2020wireless} & \multicolumn{3}{c|}{$KM+\frac{1}{4}KM^2$}             \\ \hline
			\end{tabular}
		}
	\end{center}
\end{table*}

\begin{figure}
	\centering
	\subfigure[Training overhead versus number of BS antennas $N$.]{
		\begin{minipage}[b]{0.45\textwidth}
			\includegraphics[width=3.0in]{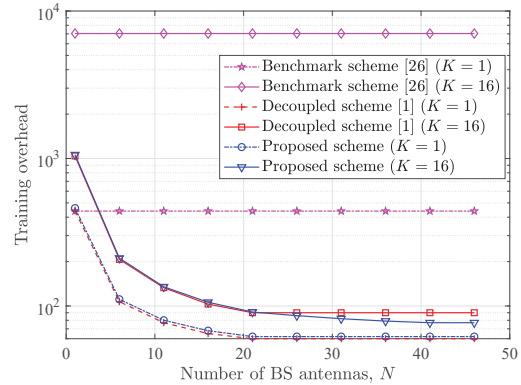}
		\end{minipage}\label{overhead_vsN}
	}\\
	\subfigure[Training overhead versus number of users $K$.]{
		\begin{minipage}[b]{0.45\textwidth}
			\includegraphics[width=3.0in]{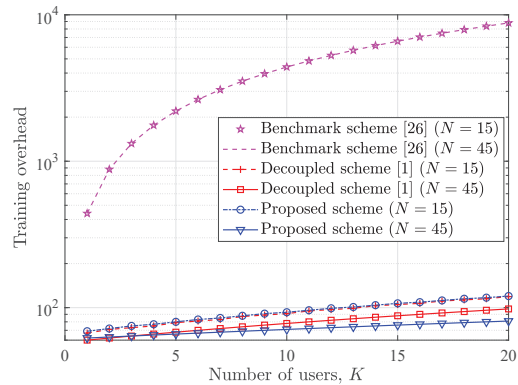}
		\end{minipage}\label{overhead_vsK}
	}
	\setlength{\abovecaptionskip}{-2pt}
	\caption{Training overhead comparison of different channel estimation schemes.} \label{overhead}
\end{figure}

In Fig.~\ref{overhead_vsN}, we show the required training overhead versus the number of BS antennas $N$. It is observed that for both the proposed and decoupled channel estimation schemes, their training overheads decrease dramatically with the increasing number of BS antennas $N$,
which is in sharp contrast to the benchmark scheme based on \cite{you2020wireless} where its training overhead is independent of $N$.
This is expected since both the proposed and decoupled channel estimation schemes exploit the (common) channel relationship and
the multiple antennas at the BS for joint cascaded channel estimation to reduce their training overheads substantially; whereas
in the benchmark scheme based on \cite{you2020wireless}, the BS estimates its cascaded channels associated with different antennas/users independently in parallel without exploiting the channel relationship between them.
When the number of BS antennas is sufficiently large (i.e., $N\ge M=40$), the required training overheads of the proposed and decoupled channel estimation schemes reach their respective lower bounds of $\frac{3}{2}M+K+1$ and
$\frac{3}{2}M+2(K-1)$ pilot symbols, as given in Table~\ref{Table of estimation}.

In Fig.~\ref{overhead_vsK}, we show the required training overhead versus the number of users $K$.
One can observe that the training overhead increment is marginal in both the proposed and decoupled channel estimation schemes as the number of users $K$ increases.
In contrast, the training overhead required by the benchmark scheme based on \cite{you2020wireless} increases dramatically with $K$ since it estimates the channels of different users separately over consecutive time.
In particular, the additional training overhead with one more user is $\text{max}\{1, \left\lceil \frac{M}{N}\right\rceil\}$ and $\text{max}\{2, 2\left\lceil \frac{M}{2N}\right\rceil\}$ for the proposed and decoupled channel estimation schemes, respectively;
whereas that for the benchmark scheme based on \cite{you2020wireless} is $M+\frac{1}{4} M^2$, which is considerably higher when $M$ is large.

\subsection{Normalized MSE Comparison for Single-User Case}
In the following simulations, we calculate the normalized MSE for different channels estimated over $1,000$ independent fading channel realizations. For example, the normalized MSE of the cascaded user$\rightarrow$IRS~1$\rightarrow$BS channel ${{\bm R}}$ estimated is given by
\begin{align}
{\varepsilon}_{{\bm R}}=\frac{1}{NM_1}     {\mathbb E}\left\{\left\| {\hat{\bm R}}-{{\bm R}}\right\|^{2}_F \Big/{ \left\|{{\bm R}}\right\|^{2}_F} \right\}.
\end{align}
The normalized MSE of other  channels estimated can be similarly calculated as the above.


\begin{figure}
	\centering
	\subfigure[Normalized MSE of ${\bar{\bm Q}}$ in Phase~I versus user transmit power $P$.]{
		\begin{minipage}[b]{0.45\textwidth}
			\includegraphics[width=3.0in]{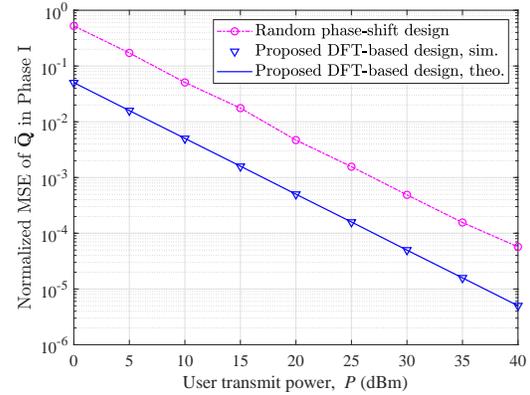}
		\end{minipage}\label{Qbar_SNR}
	}\\
	\subfigure[Normalized MSE of ${{\bm F}}$ in Phase~II versus user transmit power $P$.]{
		\begin{minipage}[b]{0.45\textwidth}
			\includegraphics[width=3.0in]{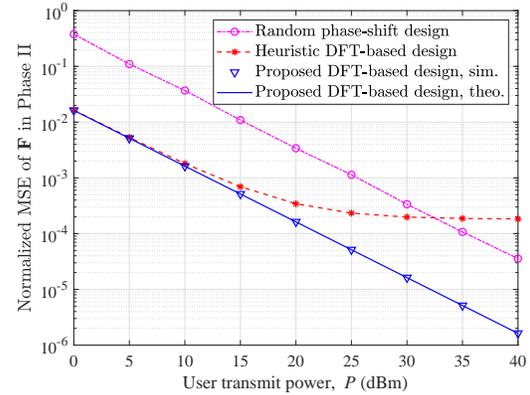}
		\end{minipage}\label{Fk_SNR}
	}
	\caption{Normalized MSE comparison of different training reflection designs for the single-user case with $N=25$.} \label{training design}
\end{figure}

We first compare different training reflection designs of the proposed channel estimation scheme under the single-user setup with $N=25$. In Figs.~\ref{Qbar_SNR} and \ref{Fk_SNR}, we show the normalized MSE versus user transmit power
$P$ for different training reflection designs in Phases I and II, respectively.
It is observed that the theoretical (theo.) analysis of MSE given in Section~\ref{training_design} is in perfect agreement with the simulation (sim.) results for our proposed optimal training design in the two Phases.
For Phase I shown in Fig.~\ref{Qbar_SNR}, the proposed DFT-based design achieves up to $10$~dB power gain over the random phase-shift design where the training reflection phase-shifts of ${\bar{\bm \Theta}}_{2,\rm I}$ are randomly generated following the uniform distribution within $[0, 2\pi)$.
For Phase II shown in Fig.~\ref{Fk_SNR}, the joint training reflection design of the two IRSs is more involved and
we consider the heuristic DFT-based design with ${\bm \Theta}_{1,\rm II} $ and ${\bm \psi}_{\rm II}$ drawn from the first (or any randomly selected) $M_1+1$ rows of the $I_2\times I_2$ DFT matrix as another benchmark design for comparison.
It is observed from Fig.~\ref{Fk_SNR} that by carefully choosing $M_1+1$ rows from the DFT matrix to ensure the perfect orthogonality of ${\bm \Omega}_{\rm II}$ with the joint training reflection design of ${\bm \Theta}_{1,\rm II} $ and ${\bm \psi}_{\rm II}$, our proposed DFT-based design achieves much lower MSE than the heuristic DFT-based design as well as the random phase-shift design.


\begin{figure}
	\centering
	\subfigure[Normalized MSE versus number of pilot symbols allocated to Phase~I for estimating $\left\{{\bar{\bm Q}}, {{\bm E}},{{\bm R}}\right\}$ in \eqref{superposed6} of Lemma~1.]{
		\begin{minipage}[b]{0.45\textwidth}
			\includegraphics[width=3.0in]{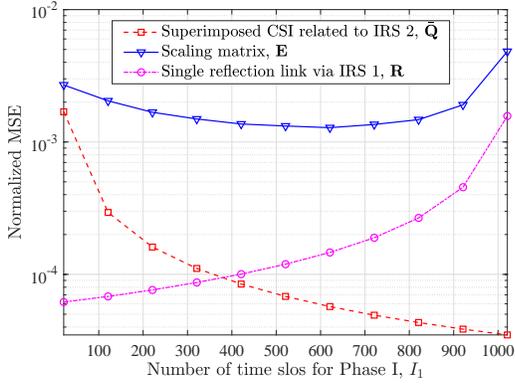}
		\end{minipage}\label{MSE_pilotN}
	}\\
	~~\subfigure[Normalized MSE versus number of pilot symbols allocated to Phase~I for estimating $\left\{\left\{ {{\bm Q}}_{m} \right\}_{m=1}^{M_1}, {\tilde{\bm R}}, {{\bm R}}\right\}$ in \eqref{superposed3}.]{
		\begin{minipage}[b]{0.45\textwidth}
			\includegraphics[width=3.0in]{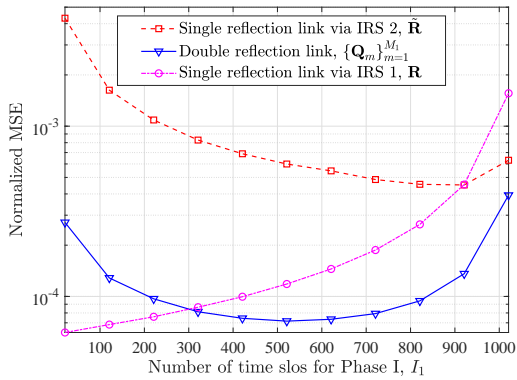}
		\end{minipage}\label{doubleRefle_pilotN}
	}
	\caption{Normalized MSE versus number of pilot symbols allocated to Phase I for the single-user case, where the total number of pilot symbols over the two phases is fixed as $I_1+I_2=1062$, $P=15$ dBm, and $N=25$.} \label{pilotN}
\end{figure}

In Fig.~\ref{pilotN}, we show the normalized MSE versus the number of pilot symbols allocated to Phase~I, $I_1$,
for estimating different CSI required for the channel models in \eqref{superposed6} and \eqref{superposed3}, respectively, where the total number of pilot symbols over the two training phases is fixed as $I_1+I_2=1062$ for the single-user case.
It is observed from Fig.~\ref{pilotN} that the MSE of $\left\{  {{\bm E}},\left\{ {{\bm Q}}_{m} \right\}_{m=1}^{M_1}, {\tilde{\bm R}}\right\}$ first decreases and then increases with the increasing number of pilot symbols allocated to Phase I.
This is expected since the channel estimation error of Phase I will affect the channel estimation performance of Phase II due to the channel estimation error propagation. Intuitively, for the estimation of $\left\{  {{\bm E}},\left\{ {{\bm Q}}_{m} \right\}_{m=1}^{M_1}, {\tilde{\bm R}}\right\}$, if more pilot symbols are allocated to
Phase I, the error propagation effect to Phase II will
be reduced, while less time is left for channel estimation in Phase~II given the fixed overall training overhead.
In contrast, as the CSI ${\bar{\bm Q}}$ (${{\bm R}}$) is estimated in Phase~I (Phase~II) only without suffering from the error propagation issue, the corresponding MSE thus monotonically decreases (increases) with the increasing number of pilot symbols allocated to Phase~I given the fixed overall training overhead.


%

As shown in Fig.~\ref{overhead}, the benchmark scheme based on \cite{you2020wireless} is much less efficient due to its higher-order training overhead.
Moreover, it was shown in \cite{zheng2020Uplink} that
given the same training overhead, the decoupled channel estimation scheme with ON/OFF IRSs already achieves much better performance than the benchmark scheme based on \cite{you2020wireless} in terms of normalized MSE.
As such, we mainly focus on the normalized MSE comparison between the proposed and decoupled channel estimation schemes in Fig.~\ref{comp_SNR}, considering their comparable minimum training overheads shown in Table~\ref{Table of estimation} with $K=1$. 

\begin{figure}
	\centering
	\subfigure[Normalized MSE versus user transmit power $P$ for estimating the CSI of the two single-reflection links $\left\{ {\tilde{\bm R}}, {{\bm R}}\right\}$.]{
		\begin{minipage}[b]{0.45\textwidth}
			\includegraphics[width=3.0in]{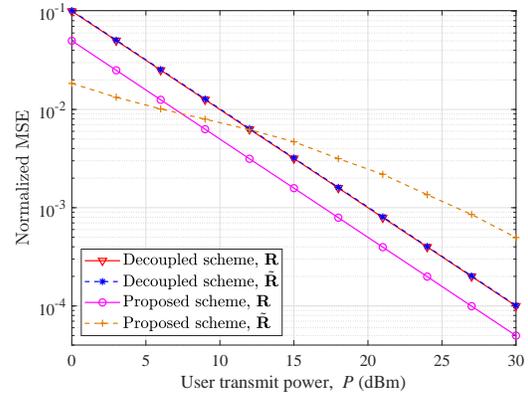}
		\end{minipage}\label{singleRefle_SNR}
	}\\
	~~\subfigure[Normalized MSE versus user transmit power $P$ for estimating the CSI of the double-reflection link $\left\{ {{\bm Q}}_{m} \right\}_{m=1}^{M_1}$.]{
		\begin{minipage}[b]{0.45\textwidth}
			\includegraphics[width=3.0in]{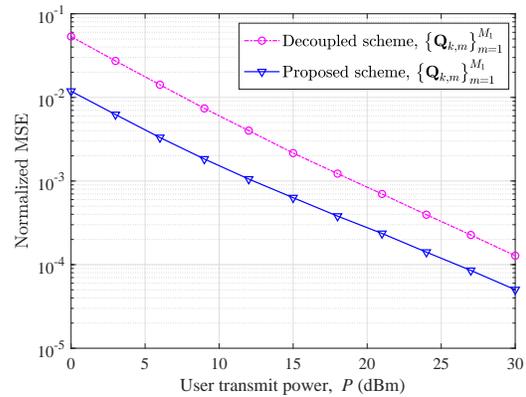}
		\end{minipage}\label{doubleRefle_SNR}
	}
	\caption{Normalized MSE comparison of different channel estimation schemes with $N=25$.} \label{comp_SNR}
\end{figure}
In Fig.~\ref{singleRefle_SNR}, we compare the normalized MSE of the proposed and decoupled channel estimation schemes
versus user transmit power $P$ for estimating the cascaded CSI of the two single-reflection links $\left\{ {\tilde{\bm R}}, {{\bm R}}\right\}$.
It is observed that the proposed channel estimation scheme has different (asymmetric) normalized MSE performance for estimating the cascaded CSI of the two single-reflection links. This can be understood by the fact that the estimation of ${\tilde{\bm R}}$ is affected by the error propagation from Phase I to Phase II; while the estimation of ${{\bm R}}$ involves Phase~II only without suffering from the error propagation issue, as discussed for Fig.~\ref{pilotN}.
In contrast, the decoupled channel estimation scheme has the same (symmetric) normalized MSE performance for estimating the cascaded CSI of the two single-reflection links, which is due to the successive/separate channel estimation of $\left\{ {\tilde{\bm R}}, {{\bm R}}\right\}$ with equal training time.
Two interesting observations are also made as follows
in comparing the performance of the proposed and decoupled channel estimation schemes.
\rev{First, by exploiting the full-reflection power of the two IRSs for estimating ${{\bm R}}$, the proposed channel estimation scheme with always-ON IRSs achieves up to 3~dB power gain over the decoupled counterpart with only one of the two IRSs turned ON successively over training.}
Second, the proposed channel estimation scheme is observed to achieve lower normalized MSE when $P\le 12$~dBm while higher normalized MSE when $P> 12$~dBm for estimating ${\tilde{\bm R}}$, as compared to the decoupled scheme.
This is because the normalized MSE performance of the proposed channel estimation scheme for estimating ${\tilde{\bm R}}$ is limited by the error propagation in the two phases, thus suffering from a lower deceasing rate of the normalized MSE
 as user transmit power $P$ increases.

In Fig.~\ref{doubleRefle_SNR}, we compare the normalized MSE of the proposed and decoupled channel estimation schemes
versus user transmit power $P$ for estimating the cascaded CSI of the double-reflection links $\left\{ {{\bm Q}}_{m} \right\}_{m=1}^{M_1}$. 
As discussed in Fig.~\ref{pilotN}, the estimation of $\left\{ {{\bm Q}}_{m} \right\}_{m=1}^{M_1}$ is affected by the error propagation from Phase I to Phase II in the proposed scheme.
\rev{On the other hand, with the signals canceled over the two single-reflection channels estimated and the estimated ${\tilde{\bm R}}$ taken as the reference CSI, the decoupled channel estimation scheme encounters not only the error propagation issue
but also the residual interference due to imperfect signal cancellation for estimating $\left\{ {{\bm Q}}_{m} \right\}_{m=1}^{M_1}$, thus resulting in higher normalized MSE than the proposed scheme.}
\subsection{Normalized MSE Comparison for Multi-User Case}
Last, we compare the normalized MSE performance between the proposed and decoupled channel estimation schemes for the multi-user case with $K=10$ in Fig.~\ref{comp_MU_SNR}, considering their comparable minimum training overheads as shown in Fig.~\ref{overhead_vsK}. Note that both the proposed and decoupled channel estimation schemes reduce the training overhead for the multi-user case by exploiting the fact that the other users' cascaded channels are lower-dimensional scaled versions of an arbitrary (reference) user's cascaded channel; as a result, both the two schemes only need to estimate the scaling vectors $\left\{{{\bm b}}_{k},{\tilde {\bm b}}_{k}\right\}_{k=2}^{K}$, with user 1 taken as the reference user.
For the decoupled channel estimation scheme,
the scaling vectors $\left\{{{\bm b}}_{k}\right\}_{k=2}^{K}$ and $\left\{{\tilde {\bm b}}_{k}\right\}_{k=2}^{K}$ are successively estimated at the BS with one of the two IRSs turned OFF;
while for the proposed channel estimation scheme, $\left\{{{\bm b}}_{k},{\tilde {\bm b}}_{k}\right\}_{k=2}^{K}$ are jointly estimated at the BS with always-ON IRSs.
\begin{figure}
	\centering
	\subfigure[Normalized MSE versus user transmit power $P$ for estimating the scaling vectors $\left\{{{\bm b}}_{k},{\tilde {\bm b}}_{k}\right\}_{k=2}^{K}$.]{
		\begin{minipage}[b]{0.45\textwidth}
			\includegraphics[width=3.0in]{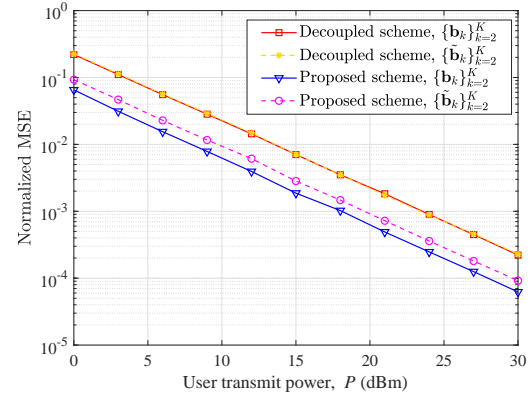}
		\end{minipage}\label{MU_SNR}
	}\\
	~~\subfigure[Normalized MSE versus user transmit power $P$ for estimating the cascaded CSI of all users.]{
		\begin{minipage}[b]{0.45\textwidth}
			\includegraphics[width=3.0in]{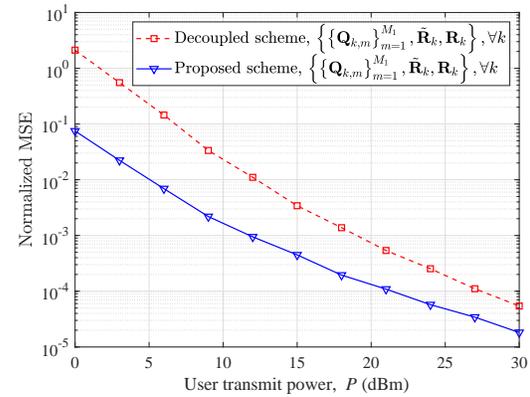}
		\end{minipage}\label{MU_SNR_reCSI}
	}
	\caption{Normalized MSE comparison of different channel estimation schemes with $N=45$ and $K=10$.} \label{comp_MU_SNR}
\end{figure}

In Fig.~\ref{MU_SNR}, we show the normalized MSE of the scaling vectors $\left\{{{\bm b}}_{k},{\tilde {\bm b}}_{k}\right\}_{k=2}^{K}$ by the proposed and decoupled channel estimation schemes
versus user transmit power $P$ under the multi-user setup, where the cascaded CSI of user 1 (i.e., $\left\{{{\bm Q}}_{1,m}\right\}_{m=1}^{M_1}$, ${\tilde{\bm R}}_{1}$, and ${{\bm R}}_{1}$) is assumed to be perfectly available as the reference CSI for ease of comparison. Some important observations are made as follows.
\rev{First, with the always-ON IRSs for maximizing the signal reflection power, the proposed channel estimation scheme achieves much better performance than the decoupled scheme with the ON/OFF IRS reflection design, in terms of the normalized MSE of $\left\{{{\bm b}}_{k}\right\}_{k=2}^{K}$ and $\left\{{\tilde {\bm b}}_{k}\right\}_{k=2}^{K}$.}
Second, the proposed channel estimation scheme has different (asymmetric) normalized MSE performance for estimating $\left\{{{\bm b}}_{k}\right\}_{k=2}^{K}$ and $\left\{{\tilde {\bm b}}_{k}\right\}_{k=2}^{K}$. This can be understood from \eqref{R_k1}-\eqref{cascaded_Q4} 
that each ${{\bm b}}_{k}$ is associated with both the single-reflection (i.e., user~$k$$\rightarrow$IRS~1$\rightarrow$BS) channel and the double-reflection (i.e., user~$k$$\rightarrow$IRS~1$\rightarrow$IRS~2$\rightarrow$BS) channel, thus reaping higher reflection power for channel estimation than each ${\tilde {\bm b}}_{k}$ that is associated with the single-reflection (i.e., user~$k$$\rightarrow$IRS~2$\rightarrow$BS) channel only.
On the other hand, the decoupled channel estimation scheme has the same (symmetric) normalized MSE performance for estimating $\left\{{{\bm b}}_{k}\right\}_{k=2}^{K}$ and $\left\{{\tilde {\bm b}}_{k}\right\}_{k=2}^{K}$, which is due to the successive/separate channel estimation of them with nearly equal training time and single-reflection power.

In Fig.~\ref{MU_SNR_reCSI}, we show the normalized MSE of the proposed and decoupled channel estimation schemes
versus user transmit power $P$ accounting for all the cascaded CSI (i.e., $\left\{\left\{ {{\bm Q}}_{k,m} \right\}_{m=1}^{M_1}, {\tilde{\bm R}}_{k}, {{\bm R}}_{k}\right\}, \forall k$) under the multi-user case.
\rev{It is observed that the proposed channel estimation scheme outperforms the decoupled counterpart significantly. This is expected due to the cumulative superior performance of the proposed channel estimation scheme as shown in Figs.~\ref{comp_SNR} and \ref{MU_SNR}.}

\section{Conclusions and Future Directions}\label{conlusion}
In this paper, we proposed an efficient uplink channel estimation scheme for the double-IRS aided multi-user MIMO system with both IRSs always turned ON during the entire channel training for maximizing their reflected signal power. 
For the single-user case, minimum training time is achieved by exploiting the fact that the cascaded CSI of the single- and double-reflection links related to IRS~2 is lower-dimensional scaled versions of their superimposed CSI.
For the multi-user case, by further exploiting the fact that all other users' cascaded channels are also lower-dimensional scaled versions of an arbitrary (reference) user's cascaded channel (estimated as in the single-user case), the multi-user cascaded channels are estimated with minimum training overhead. 
Moreover, for the proposed channel estimation scheme, we designed the corresponding training reflection phase-shifts for the double IRSs to minimize the channel estimation error over different training phases.
In addition, we showed an interesting trade-off of training time allocation over different training phases with their intricate error propagation effect taken into account.
Simulation results demonstrated the effectiveness of the proposed channel estimation scheme and training reflection phase-shift design,
as compared to the existing channel estimation schemes and other heuristic training designs.

\rev{In the future, it is an interesting direction to study the more general multi-IRS aided multi-user communication system with multi-hop signal reflection \cite{mei2021multi,you2020deploy}, which calls for more technically challenging system designs on channel estimation, joint passive beamforming, and multi-IRS deployment.
Moreover, besides the CSI-based (robust) beamforming design, the beam training \cite{you2020fast} and machine learning \cite{kim2020exploiting} based methods (that do not require the CSI explicitly) are also promising solutions to the multi-IRS aided multi-user communication system, which deserve more investigation in the future.}

\ifCLASSOPTIONcaptionsoff
  \newpage
\fi

\bibliographystyle{IEEEtran}
\bibliography{IRS_MIMO_CE}

\begin{thebibliography}{10}
\providecommand{\url}[1]{#1}
\csname url@samestyle\endcsname
\providecommand{\newblock}{\relax}
\providecommand{\bibinfo}[2]{#2}
\providecommand{\BIBentrySTDinterwordspacing}{\spaceskip=0pt\relax}
\providecommand{\BIBentryALTinterwordstretchfactor}{4}
\providecommand{\BIBentryALTinterwordspacing}{\spaceskip=\fontdimen2\font plus
\BIBentryALTinterwordstretchfactor\fontdimen3\font minus
  \fontdimen4\font\relax}
\providecommand{\BIBforeignlanguage}[2]{{%
\expandafter\ifx\csname l@#1\endcsname\relax
\typeout{** WARNING: IEEEtran.bst: No hyphenation pattern has been}%
\typeout{** loaded for the language `#1'. Using the pattern for}%
\typeout{** the default language instead.}%
\else
\language=\csname l@#1\endcsname
\fi
#2}}
\providecommand{\BIBdecl}{\relax}
\BIBdecl

\bibitem{zheng2020Uplink}
B.~Zheng, C.~You, and R.~Zhang, ``Uplink channel estimation for double-{IRS}
  assisted multi-user {MIMO},'' in \emph{Proc. IEEE Int. Conf. Commun. (ICC)},
  Montreal, Canada, Jun. 2020, pp. 1--6, accepted.

\bibitem{wu2020intelligent}
Q.~Wu, S.~Zhang, B.~Zheng, C.~You, and R.~Zhang, ``Intelligent reflecting
  surface aided wireless communications: A tutorial,'' \emph{IEEE Trans.
  Commun.}, doi: 10.1109/TCOMM.2021.3051897, Jan. 2021.

\bibitem{qingqing2019towards}
Q.~Wu and R.~Zhang, ``Towards smart and reconfigurable environment: Intelligent
  reflecting surface aided wireless network,'' \emph{IEEE Commun. Mag.},
  vol.~58, no.~1, pp. 106--112, Jan. 2020.

\bibitem{Renzo2019Smart}
M.~Di~Renzo \emph{et~al.}, ``Smart radio environments empowered by
  reconfigurable {AI} meta-surfaces: An idea whose time has come,''
  \emph{EURASIP J. Wireless Commun. Netw.}, vol. 2019:129, May 2019.

\bibitem{basar2019wireless}
E.~Basar, M.~Di~Renzo, J.~de~Rosny, M.~Debbah, M.-S. Alouini, and R.~Zhang,
  ``Wireless communications through reconfigurable intelligent surfaces,''
  \emph{IEEE Access}, vol.~7, pp. 116\,753--116\,773, Aug. 2019.

\bibitem{Zheng2020IRSNOMA}
B.~{Zheng}, Q.~{Wu}, and R.~{Zhang}, ``Intelligent reflecting surface-assisted
  multiple access with user pairing: {NOMA or OMA}?'' \emph{IEEE Commun.
  Lett.}, vol.~24, no.~4, pp. 753--757, Apr. 2020.

\bibitem{Pan2020Multicell}
C.~{Pan}, H.~{Ren}, K.~{Wang}, W.~{Xu}, M.~{Elkashlan}, A.~{Nallanathan}, and
  L.~{Hanzo}, ``Multicell {MIMO} communications relying on intelligent
  reflecting surfaces,'' \emph{IEEE Trans. Wireless Commun.}, vol.~19, no.~8,
  pp. 5218--5233, Aug. 2020.

\bibitem{Yang2020IRS}
Y.~{Yang}, S.~{Zhang}, and R.~{Zhang}, ``{IRS}-enhanced {OFDMA}: Joint resource
  allocation and passive beamforming optimization,'' \emph{IEEE Wireless
  Commun. Lett.}, vol.~9, no.~6, pp. 760--764, Jun. 2020.

\bibitem{Wu2020Weighted2}
Q.~{Wu} and R.~{Zhang}, ``Joint active and passive beamforming optimization for
  intelligent reflecting surface assisted {SWIPT} under {QoS} constraints,''
  \emph{IEEE J. Sel. Areas Commun.}, vol.~38, no.~8, pp. 1735--1748, Aug. 2020.

\bibitem{Pan2020Intelligent}
C.~{Pan}, H.~{Ren}, K.~{Wang}, M.~{Elkashlan}, A.~{Nallanathan}, J.~{Wang}, and
  L.~{Hanzo}, ``Intelligent reflecting surface aided {MIMO} broadcasting for
  simultaneous wireless information and power transfer,'' \emph{IEEE J. Sel.
  Areas Commun.}, vol.~38, no.~8, pp. 1719--1734, Aug. 2020.

\bibitem{Cui2019Secure}
M.~{Cui}, G.~{Zhang}, and R.~{Zhang}, ``Secure wireless communication via
  intelligent reflecting surface,'' \emph{IEEE Wireless Commun. Lett.}, vol.~8,
  no.~5, pp. 1410--1414, Oct. 2019.

\bibitem{Zhou2020Delay}
F.~{Zhou}, C.~{You}, and R.~{Zhang}, ``Delay-optimal scheduling for {IRS}-aided
  mobile edge computing,'' \emph{IEEE Wireless Commun. Lett.}, doi:
  10.1109/LWC.2020.3042189, Dec. 2020.

\bibitem{Wu2019TWC}
Q.~Wu and R.~Zhang, ``Intelligent reflecting surface enhanced wireless network
  via joint active and passive beamforming,'' \emph{IEEE Trans. Wireless
  Commun.}, vol.~18, no.~11, pp. 5394--5409, Nov. 2019.

\bibitem{Huang2019Reconfigurable}
C.~{Huang}, A.~{Zappone}, G.~C. {Alexandropoulos}, M.~{Debbah}, and C.~{Yuen},
  ``Reconfigurable intelligent surfaces for energy efficiency in wireless
  communication,'' \emph{IEEE Trans. Wireless Commun.}, vol.~18, no.~8, pp.
  4157--4170, Aug. 2019.

\bibitem{wu2019beamforming}
Q.~Wu and R.~Zhang, ``Beamforming optimization for wireless network aided by
  intelligent reflecting surface with discrete phase shifts,'' \emph{IEEE
  Trans. Commun.}, vol.~68, no.~3, pp. 1838--1851, Mar. 2020.

\bibitem{mei2020performance}
W.~Mei and R.~Zhang, ``Performance analysis and user association optimization
  for wireless network aided by multiple intelligent reflecting surfaces,''
  \emph{arXiv preprint arXiv:2009.02551}, 2020.

\bibitem{Zhang2020Intelligent}
S.~{Zhang} and R.~{Zhang}, ``Intelligent reflecting surface aided multiple
  access: Capacity region and deployment strategy,'' in \emph{Proc. IEEE Int.
  Wkshps. Signal Process. Adv. Wireless Commun. (SPAWC)}, Atlanta, GA, USA, May
  2020, pp. 1--5.

\bibitem{yang2020energy}
Z.~Yang, M.~Chen, W.~Saad, W.~Xu, M.~Shikh-Bahaei, H.~V. Poor, and S.~Cui,
  ``Energy-efficient wireless communications with distributed reconfigurable
  intelligent surfaces,'' \emph{arXiv preprint arXiv:2005.00269}, 2020.

\bibitem{Zhang2020Capacity}
Z.~{Zhang} and L.~{Dai}, ``Capacity improvement in wideband reconfigurable
  intelligent surface-aided cell-free network,'' in \emph{Proc. IEEE Int.
  Wkshps. Signal Process. Adv. Wireless Commun. (SPAWC)}, Atlanta, GA, USA, May
  2020, pp. 1--5.

\bibitem{li2019jointactive}
X.~Li, J.~Fang, F.~Gao, and H.~Li, ``Joint active and passive beamforming for
  intelligent reflecting surface-assisted massive {MIMO} systems,'' \emph{arXiv
  preprint arXiv:1912.00728}, 2019.

\bibitem{yang2020outage}
L.~Yang, Y.~Yang, D.~B. da~Costa, and I.~Trigui, ``Outage probability and
  capacity scaling law of multiple {RIS}-aided cooperative networks,''
  \emph{arXiv preprint arXiv:2007.13293}, 2020.

\bibitem{Gao2020Distributed}
Y.~{Gao}, J.~{Xu}, W.~{Xu}, D.~W.~K. {Ng}, and M.~S. {Alouini}, ``Distributed
  {IRS} with statistical passive beamforming for {MISO} communications,''
  \emph{IEEE Wireless Commun. Lett.}, vol.~10, no.~2, pp. 221--225, Feb. 2021.

\bibitem{Wei2020Sum}
Z.~{Wei}, Y.~{Cai}, Z.~{Sun}, D.~W.~K. {Ng}, J.~{Yuan}, M.~{Zhou}, and
  L.~{Sun}, ``Sum-rate maximization for {IRS}-assisted {UAV OFDMA}
  communication systems,'' \emph{IEEE Trans. Wireless Commun.}, doi:
  10.1109/TWC.2020.3042977, Dec. 2020.

\bibitem{Han2020Cooperative}
Y.~{Han}, S.~{Zhang}, L.~{Duan}, and R.~{Zhang}, ``Cooperative double-{IRS}
  aided communication: Beamforming design and power scaling,'' \emph{IEEE
  Wireless Commun. Lett.}, vol.~9, no.~8, pp. 1206--1210, Aug. 2020.

\bibitem{Zheng2020DoubleIRS}
B.~Zheng, C.~You, and R.~Zhang, ``Double-{IRS} assisted multi-user {MIMO}:
  Cooperative passive beamforming design,'' \emph{IEEE Trans. Wireless
  Commun.}, doi: 10.1109/TWC.2021.3059945, Feb. 2021.

\bibitem{you2020wireless}
C.~You, B.~Zheng, and R.~Zhang, ``Wireless communication via double {IRS}:
  Channel estimation and passive beamforming designs,'' \emph{IEEE Wireless
  Commun. Lett.}, vol.~10, no.~2, pp. 431--435, Feb. 2021.

\bibitem{zheng2019intelligent}
B.~Zheng and R.~Zhang, ``Intelligent reflecting surface-enhanced {OFDM}:
  Channel estimation and reflection optimization,'' \emph{IEEE Wireless Commun.
  Lett.}, vol.~9, no.~4, pp. 518--522, Apr. 2020.

\bibitem{zheng2020intelligent}
B.~Zheng, C.~You, and R.~Zhang, ``Intelligent reflecting surface assisted
  multi-user {OFDMA}: Channel estimation and training design,'' \emph{IEEE
  Trans. Wireless Commun.}, vol.~19, no.~12, pp. 8315--8329, Dec. 2020.

\bibitem{zheng2020fast}
------, ``Fast channel estimation for {IRS}-assisted {OFDM},'' \emph{IEEE
  Wireless Commun. Lett.}, doi: 10.1109/LWC.2020.3038434, Nov. 2020.

\bibitem{you2019progressive}
C.~You, B.~Zheng, and R.~Zhang, ``Channel estimation and passive beamforming
  for intelligent reflecting surface: Discrete phase shift and progressive
  refinement,'' \emph{IEEE J. Sel. Areas Commun.}, vol.~38, no.~11, pp.
  2604--2620, Nov. 2020.

\bibitem{you2020fast}
------, ``Fast beam training for {IRS}-assisted multiuser communications,''
  \emph{IEEE Wireless Commun. Lett.}, vol.~9, no.~11, pp. 1845--1849, Nov.
  2020.

\bibitem{wang2019channel}
Z.~{Wang}, L.~{Liu}, and S.~{Cui}, ``Channel estimation for intelligent
  reflecting surface assisted multiuser communications: Framework, algorithms,
  and analysis,'' \emph{IEEE Trans. Wireless Commun.}, vol.~19, no.~10, pp.
  6607--6620, Oct. 2020.

\bibitem{jensen2019optimal}
T.~L. Jensen and E.~De~Carvalho, ``An optimal channel estimation scheme for
  intelligent reflecting surfaces based on a minimum variance unbiased
  estimator,'' in \emph{Proc. IEEE Int. Conf. Acoust., Speech, Signal Process.
  (ICASSP)}, Barcelona, Spain, May 2020, pp. 5000--5004.

\bibitem{liu2019matrix}
H.~{Liu}, X.~{Yuan}, and Y.~J. {Zhang}, ``Matrix-calibration-based cascaded
  channel estimation for reconfigurable intelligent surface assisted multiuser
  {MIMO},'' \emph{IEEE J. Sel. Areas Commun.}, vol.~38, no.~11, pp. 2621--2636,
  Nov. 2020.

\bibitem{wang2019compressed}
P.~Wang, J.~Fang, H.~Duan, and H.~Li, ``Compressed channel estimation for
  intelligent reflecting surface-assisted millimeter wave systems,'' \emph{IEEE
  Signal Process. Lett.}, vol.~27, pp. 905--909, May 2020.

\bibitem{Zhang2020Cascaded}
W.~{Zhang}, J.~{Xu}, W.~{Xu}, D.~W.~K. {Ng}, and H.~{Sun}, ``Cascaded channel
  estimation for {IRS}-assisted mmwave multi-antenna with quantized
  beamforming,'' \emph{IEEE Commun. Lett.}, vol.~25, no.~2, pp. 593--597, Feb.
  2021.

\bibitem{yang2019intelligent}
Y.~Yang, B.~Zheng, S.~Zhang, and R.~Zhang, ``Intelligent reflecting surface
  meets {OFDM}: Protocol design and rate maximization,'' \emph{IEEE Trans.
  Commun.}, vol.~68, no.~7, pp. 4522--4535, Jul. 2020.

\bibitem{you2020deploy}
C.~You, B.~Zheng, and R.~Zhang, ``How to deploy intelligent reflecting surfaces
  in wireless network: {BS}-side, user-side, or both sides?'' \emph{arXiv
  preprint arXiv:2012.03403}, 2020.

\bibitem{Polyphase1972Polyphase}
D.~{Chu}, ``Polyphase codes with good periodic correlation properties,''
  \emph{IEEE Trans. Inf. Theory}, vol.~18, no.~4, pp. 531--532, Jul. 1972.

\bibitem{mei2021multi}
W.~Mei and R.~Zhang, ``Multi-beam multi-hop routing for intelligent reflecting
  surfaces aided massive {MIMO},'' \emph{arXiv preprint arXiv:2101.00217},
  2021.

\bibitem{kim2020exploiting}
J.~Kim, S.~Hosseinalipour, T.~Kim, D.~J. Love, and C.~G. Brinton, ``Exploiting
  multiple intelligent reflecting surfaces in multi-cell uplink {MIMO}
  communications,'' \emph{arXiv preprint arXiv:2011.01141}, 2020.

\end{thebibliography}

\begin{IEEEbiography}[{\includegraphics[width=1in,height=1.25in,clip,keepaspectratio]{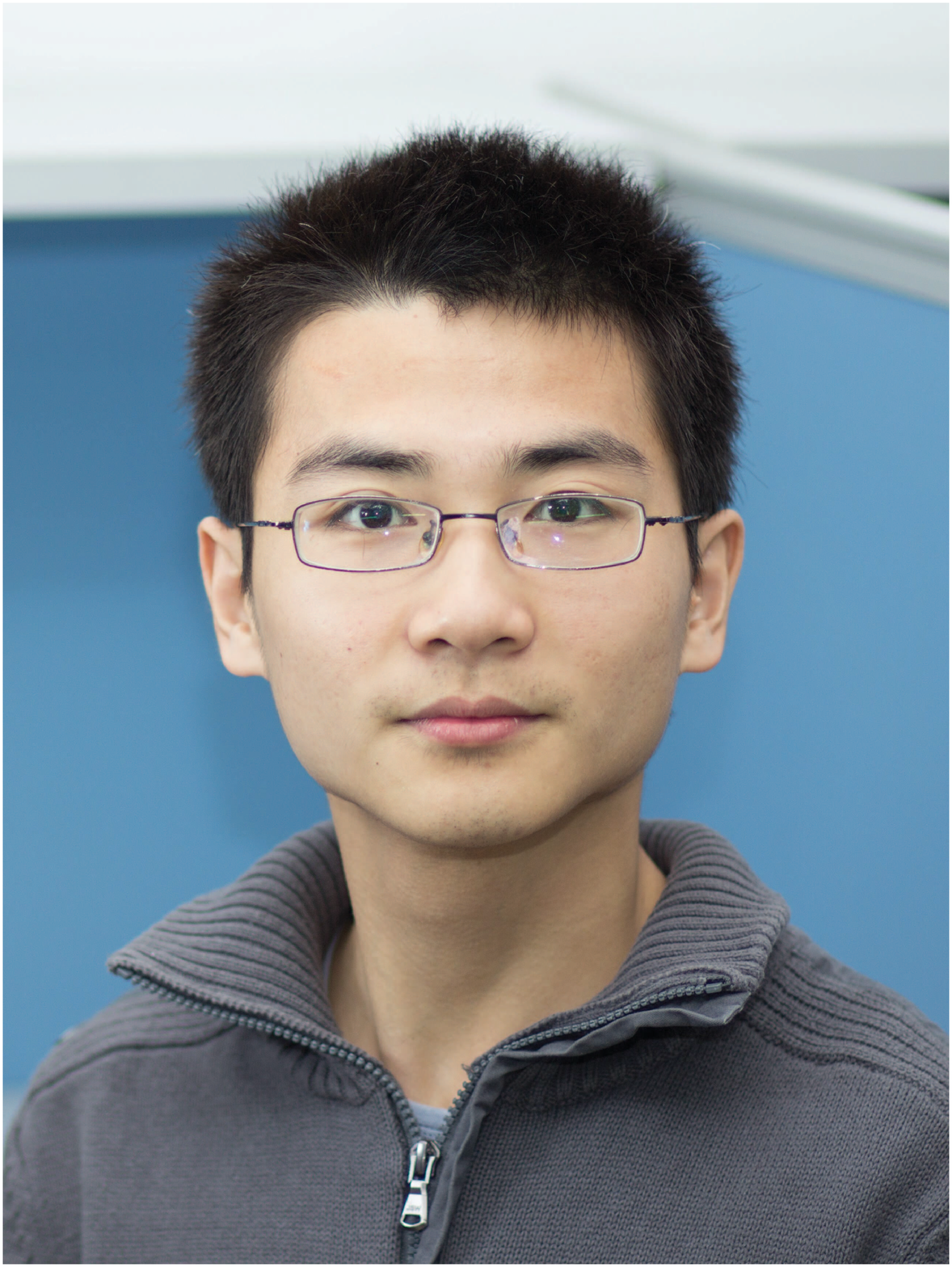}}]{Beixiong Zheng}
	(M'18) received the B.S. and Ph.D.
	degrees from the South China University of Technology, Guangzhou, China, in 2013 and 2018, respectively. 
	He is currently a Research Fellow with the Department of Electrical and Computer Engineering, National University of Singapore.
	His recent research interests include intelligent reflecting surface (IRS), index modulation (IM), and non-orthogonal multiple access (NOMA). 
	
	From 2015 to 2016, he was a Visiting Student Research Collaborator with Columbia University, New York, NY, USA. He was a recipient of
	the Best Paper Award from the IEEE International Conference on Computing, Networking and Communications in 2016, 
	the Best Ph.D. Thesis Award from China Education Society of Electronics in 2018,
	the Exemplary Reviewer of the IEEE TRANSACTIONS ON COMMUNICATIONS,
	and the Outstanding Reviewer of Physical Communication.
\end{IEEEbiography}

\begin{IEEEbiography}[{\includegraphics[width=1in,height=1.25in,clip,keepaspectratio]{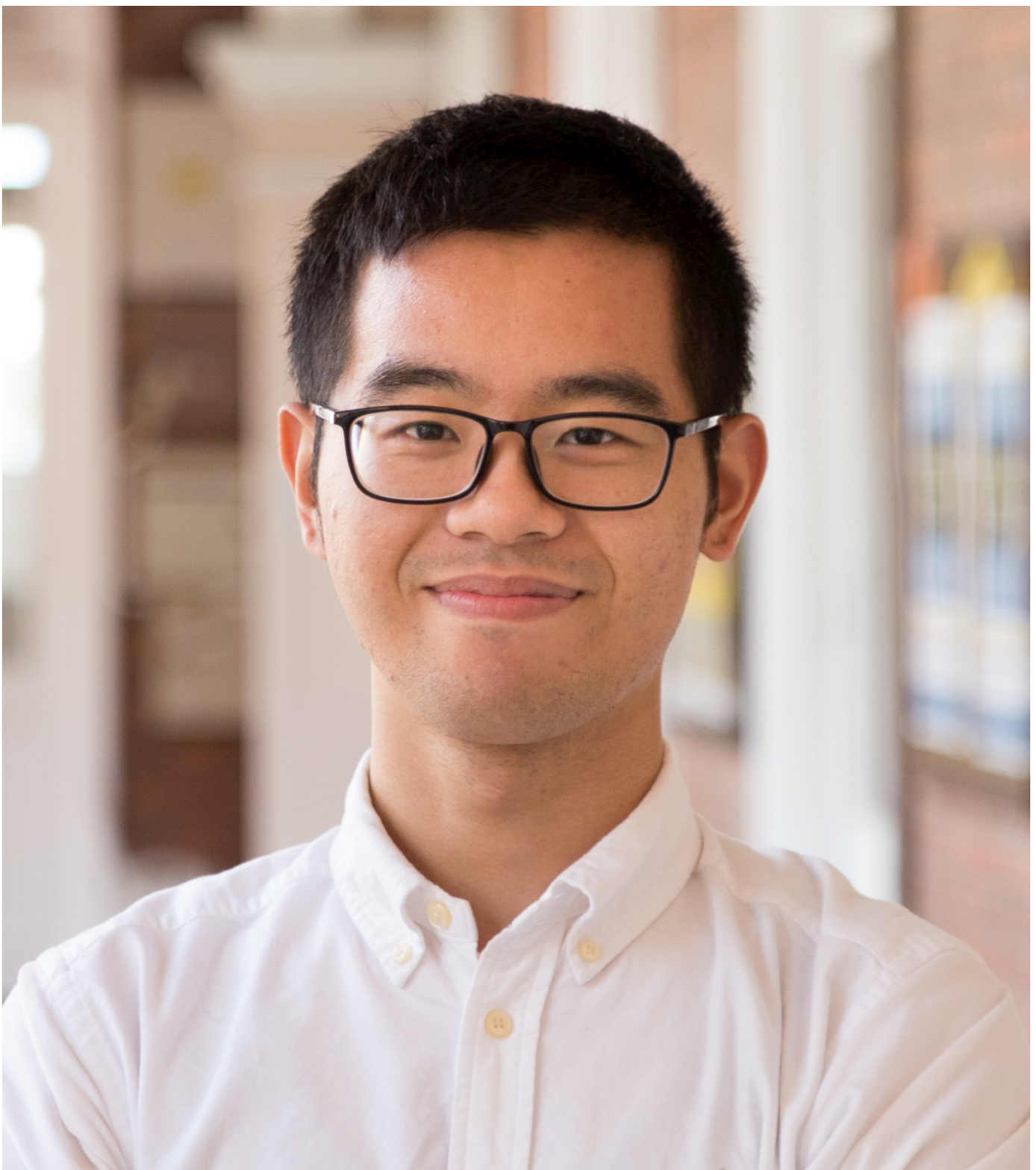}}]{Changsheng You}(M'19) received his B.Eng. degree in 2014 from University of Science and Technology of China in electronic engineering and information science, and Ph.D. degree in 2018 from The University of Hong Kong in electrical and electronic engineering. He is currently a Research Fellow with the Department of Electrical and Computer Engineering, National University of Singapore. He is an editor for IEEE COMMUNICATIONS LETTERS. His research interests include intelligent reflecting surface, UAV communications, edge learning, mobile-edge computing, wireless power transfer, and convex optimization.
	
	Dr. You received the IEEE Communications Society Asia-Pacific Region Outstanding Paper Award in 2019 and the Exemplary Reviewer of the IEEE TRANSACTIONS ON COMMUNICATIONS and IEEE TRANSACTIONS ON WIRELESS COMMUNICATIONS. He was a recipient of the Best Ph.D. Thesis Award of The University of Hong Kong in 2019.
\end{IEEEbiography}

\begin{IEEEbiography}[{\includegraphics[width=1in,height=1.25in,clip,keepaspectratio]{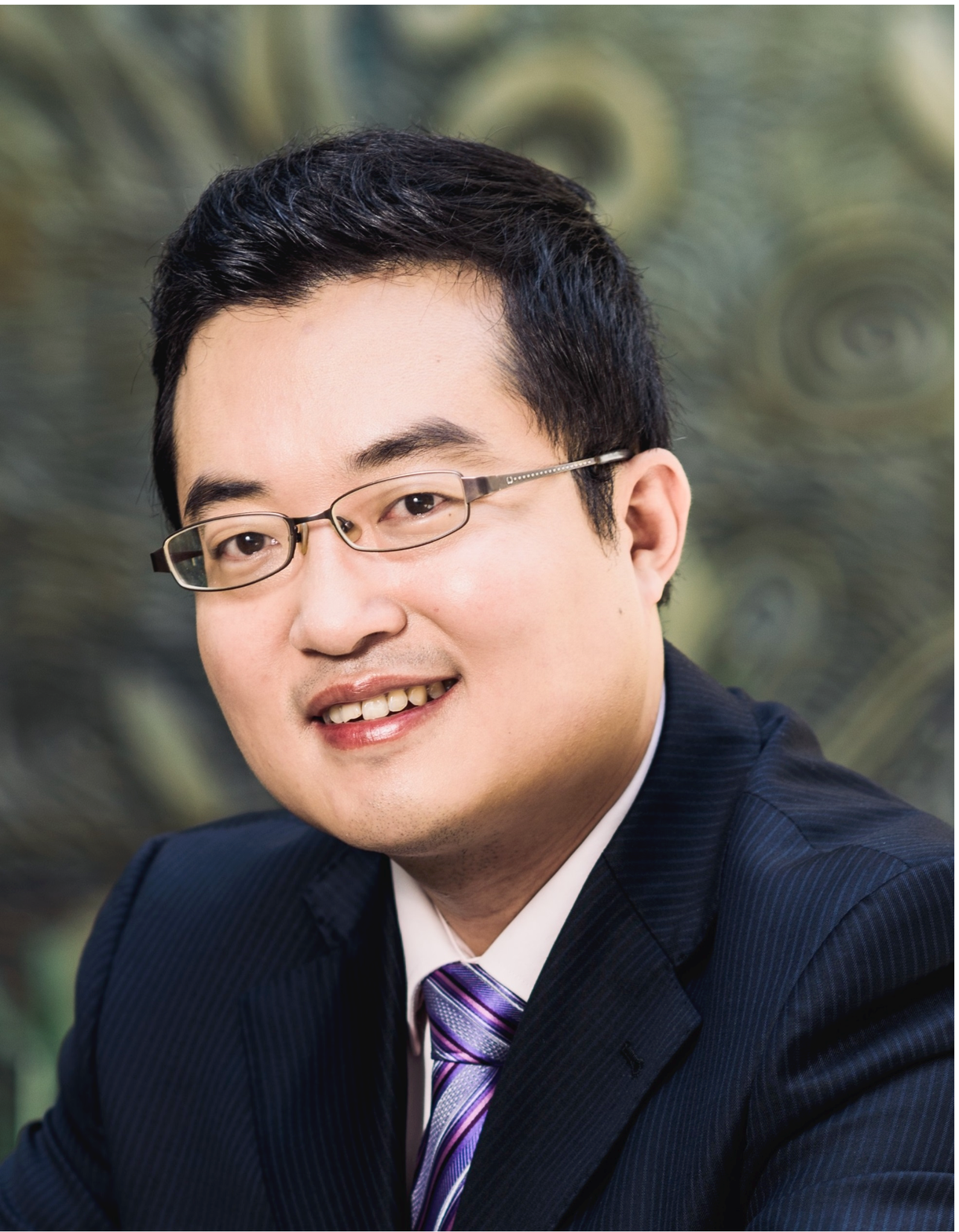}}] {Rui Zhang} (S'00-M'07-SM'15-F'17) received the B.Eng. (first-class Hons.) and M.Eng. degrees from the National University of Singapore, Singapore, and the Ph.D. degree from the Stanford University, Stanford, CA, USA, all in electrical engineering.
	
	From 2007 to 2010, he worked at the Institute for Infocomm Research, ASTAR, Singapore. Since 2010, he has been working with the National University of Singapore, where he is now a Professor in the Department of Electrical and Computer Engineering. He has published over 200 journal papers and over 180 conference papers. He has been listed as a Highly Cited Researcher by Thomson Reuters/Clarivate Analytics since 2015. His current research interests include UAV/satellite communications, wireless power transfer, reconfigurable MIMO, and optimization methods.     
	
	He was the recipient of the 6th IEEE Communications Society Asia-Pacific Region Best Young Researcher Award in 2011, the Young Researcher Award of National University of Singapore in 2015, and the Wireless Communications Technical Committee Recognition Award in 2020. He was the co-recipient of the IEEE Marconi Prize Paper Award in Wireless Communications in 2015 and 2020, the IEEE Communications Society Asia-Pacific Region Best Paper Award in 2016, the IEEE Signal Processing Society Best Paper Award in 2016, the IEEE Communications Society Heinrich Hertz Prize Paper Award in 2017 and 2020, the IEEE Signal Processing Society Donald G. Fink Overview Paper Award in 2017, and the IEEE Technical Committee on Green Communications \& Computing (TCGCC) Best Journal Paper Award in 2017. His co-authored paper received the IEEE Signal Processing Society Young Author Best Paper Award in 2017. He served for over 30 international conferences as the TPC co-chair or an organizing committee member, and as the guest editor for 3 special issues in the IEEE JOURNAL OF SELECTED TOPICS IN SIGNAL PROCESSING and the IEEE JOURNAL ON SELECTED AREAS IN COMMUNICATIONS. He was an elected member of the IEEE Signal Processing Society SPCOM Technical Committee from 2012 to 2017 and SAM Technical Committee from 2013 to 2015, and served as the Vice Chair of the IEEE Communications Society Asia-Pacific Board Technical Affairs Committee from 2014 to 2015. He served as an Editor for the IEEE TRANSACTIONS ON WIRELESS COMMUNICATIONS from 2012 to 2016, the IEEE JOURNAL ON SELECTED AREAS IN COMMUNICATIONS: Green Communications and Networking Series from 2015 to 2016, the IEEE TRANSACTIONS ON SIGNAL PROCESSING from 2013 to 2017, and the IEEE TRANSACTIONS ON GREEN COMMUNICATIONS AND NETWORKING from 2016 to 2020. He is now an Editor for the IEEE TRANSACTIONS ON COMMUNICATIONS. He serves as a member of the Steering Committee of the IEEE Wireless Communications Letters. He is a Distinguished Lecturer of IEEE Signal Processing Society and IEEE Communications Society.
\end{IEEEbiography}

\end{document}